\documentclass{lmcs}
\pdfoutput=1

\usepackage{lastpage}
\lmcsdoi{15}{3}{10}
\lmcsheading{}{\pageref{LastPage}}{}{}%
{Dec.~29,~2017}{Aug.~01,~2019}{}


\usepackage{myllncs,  mymath,  proof,  latexsym,  amssymb, hyperref}

\usepackage{color}
\usepackage{graphicx}
\usepackage{rotating}

\long\def\Makoto#1{{\small{\color{red} Makoto: {\em #1}}}}
\long\def\Stefano#1{{\small{\color{Green} Stefano: {\em #1}}}}
\long\def\Makoto#1{}
\long\def\Stefano#1{}

\newif\ifpaperlong
   \paperlongtrue

\begin{document}
\sloppy \hbadness=10000 \vbadness=10000

\long\def\Erase#1{}

% macros for LMCS by Makoto

\def\calH{{\mathcal H}}
\def\calU{{\mathcal U}}
\def\calR{{\mathcal R}}
\def\calM{{\mathcal M}}
\def\calD{{\mathcal D}}
\def\Codom{{\rm codom}}
\def\Id{{\rm id}}

% macros for fossacs by Makoto

%\def\s{{\bf S}}
\def\s{{ s}}

\def\CLKIDomega{\CLKID}
\def\Dom{{\rm dom}}
\def\Codom{{\rm codom}}

% instead of amsmath

%\def\implies{\Longrightarrow}
\def\implies{\rightarrow}

% Stefano macros (2016.9)

% Author macros %%%%%%%%%%%%%%%%%%%%%%%%%%%%%%%%%%%%%%%%%%%%%%%%
%
\newcommand{\setA}              { {\mathcal A} }
\newcommand{\setB}              { {\mathcal B} }
\newcommand{\setC}              { {\mathcal C} }
\newcommand{\setD}              { {\mathcal D} }
\newcommand{\setE}              { {\mathcal E} }
\newcommand{\setF}              { {\mathcal F} }
\newcommand{\setG}              { {\mathcal G} }
\newcommand{\setH}              { {\mathcal H} }
\newcommand{\setI}              { {\mathcal I} }
\newcommand{\setL}              { {\mathcal L} }
\newcommand{\setM}              { {\mathcal M} }
\newcommand{\setN}              { {\mathcal N} }
\newcommand{\setO}              { {\mathcal O} }
\newcommand{\setP}              { {\mathcal P} }
\newcommand{\setQ}              { {\mathcal Q} }
\newcommand{\setR}              { {\mathcal R} }
\newcommand{\setS}              { {\mathcal S} }
\newcommand{\setT}              { {\mathcal T} }
\newcommand{\setU}              { {\mathcal U} }
\newcommand{\setW}              { {\mathcal W} }
\newcommand{\setX}              { {\mathcal X} }
\newcommand{\setY}              { {\mathcal Y} }
\newcommand{\setZ}              { {\mathcal Z} }
\newcommand{\Nat}               { {\mathbb{N}}  }
\newcommand{\Npredicate}        { N }

\newcommand{\PA}                { {\tt PA}}
\newcommand{\false}             { {\tt false} }
\newcommand{\true }             { {\tt true } }
\newcommand{\WF}                { {\tt WF} }
\newcommand{\length}            { {\tt l} }
\newcommand{\Subtree}           { {\tt SubT} }
\newcommand{\List}            { {\tt List} }
\newcommand{\FV}              { {\tt FV} }
\newcommand{\Occ}             { {\tt Occ} }
\newcommand{\id}              { {\tt id} }

\newcommand{\comp}            { {\circ} }

\newcommand{\Comp}            { {\circ} }
\newcommand{\dom}             { {\tt dom} }
\newcommand{\cons}            { {\tt cons} }
\newcommand{\Tree}            { {\tt Tree} }
\newcommand{\Lab}             { {\tt Lab} }
\newcommand{\EM}              { {\tt EM} }
\newcommand{\restr}           { {\lceil} }
\newcommand{\Ind}              { {\tt Ind} }
\newcommand{\FOInd}            { {\tt FOInd} }
\newcommand{\IndR}            { {\tt IndR} }
\newcommand{\Bool}            { {\tt Bool} }
\newcommand{\lexicographic}   { {  \prec  } }
\newcommand{\nil}             { {\tt nil} }
\newcommand{\Ch}              { {\tt Ch} }
\newcommand{\Bud}              { {\tt Bud} }
\newcommand{\implied}          { {\Leftarrow} }

%NEW MACROS REQUIRED BY BROTHERSTON: 06/10/2018
\newcommand{\Zeta}             { {\mathbb{Z}} }
\newcommand{\zeroZ}            { {0_{\Zeta}}}
\newcommand{\NatModel}               { {\tt N} }
\newcommand{\ZetaModelMinus}      { {{\tt Z}^{-}} }
\newcommand{\ZetaModelPlus}      { {{\tt Z}^{+}_0} }
\newcommand{\universe}            [1] {   { |\!\!| #1 |\!\!| } }
%END: NEW MACROS REQUIRED BY BROTHERSTON: 06/10/2018

\newcommand{\range}                    { {\tt range} }
\newcommand{\domain}             [2] {#2(#1)}

\newcommand{\LJID}             { {\tt LJID} }
\newcommand{\LKID}             { {\tt LKID} }
\newcommand{\CLKID}            { {\tt CLKID}^\omega }
\newcommand{\card}             { {\tt card} }
\newcommand{\Rational}         { {\mathbb Q} }
\newcommand{\subsetsim}        { { {}^{\subset}_{\sim } } }
\newcommand{\codom}            { {\tt codom} }
\newcommand{\zeroaxiom}        { { 0\mbox{-axiom} } }

% instead of amsmath

\def\text#1{{\rm #1}}
\def\Candidates{\text{Candidates}}

% Presburger

\newcommand{\code}[1]{{{\ensuremath{\tt #1}}}}
\def\croot{\code{root}}
\def\ctrue{\code{true}}
\def\csortll{\code{sortll}}
\def\Sortll{\code{sortll}}
\def\Avl{\code{avl}}
\def\Sortavl{\code{sortavl}}

\def\NC{{\hbox{NC}}}
\def\PI{\hbox{PI}}
\def\DPI{\hbox{DPI}}
\def\SL{\hbox{SL}}
\def\SLA#1{\hbox{SLA{#1}}}
\def\Null{{\hbox{null}}}
\def\Pred{{\hbox{\tt pred\ }}}
\def\Z{{\hbox{\bf Z}}}
\def\ML{\code{ML}}
\def\NML{\code{NML}}
\def\True{\code{true}}
\def\False{\code{false}}
\def\Mod{\hbox{mod\ }}
\def\ProjS{{\cal P}_S}
\def\ProjN{{\cal P}_N}

% revise

\def\Nll{\hbox{nll}}

% btw

\def\SLRDbtw{\hbox{SLRD}_{btw}}
\def\BTW{{\hbox{BTW}}}
\def\Loc{{\hbox{Loc}}}
\def\Ne{{\ne}}

% weak establish

\def\Connected{\hbox{Connected}}
\def\Eststablished{\hbox{Eststablished}}
\def\Valued{\hbox{Valued}}

% monadic

\def\Nil{{\hbox{nil}}}
\def\eqDef{=_{\hbox{def}}}
\def\Inf#1{{\infty_{#1}}}
\def\Stores{\hbox{Stores}}
\def\SVars{{\hbox{SVars}}}
\def\Val{{\hbox{Val}}}
\def\MSO{{\hbox{MSO}}}
\def\Sep{{\hbox{Sep}}}
\def\THeaps{\hbox{THeaps}}
\def\Root{\hbox{Root}}
\def\TG{\hbox{TGraph}}

\def\Tilde{\widetilde}
\def\Bar{\overline}
\def\Lequiv{\Longleftrightarrow}
\def\Lto{\Longrightarrow}
\def\Lfrom{\Longleftarrow}
\def\Norm{\hbox{Norm}}
\def\Noshare{\hbox{Noshare}}
\def\Roots{\hbox{Roots}}
\def\Forest{\hbox{Forest}}
\def\Var{\hbox{Var}}
\def\Range{\hbox{Range}}
\def\Cell{\hbox{Cell}}
\def\Tree{\hbox{Tree}}
\def\Switch{\hbox{Switch}}
\def\All{\hbox{All}}
\def\To{\leadsto}
\def\tree{\hbox{tree}}

% previsou

\long\def\J#1{} % skip
\def\T#1{\hbox{\color{green}{$\clubsuit #1$}}}
\def\W#1{\hbox{\color{Orange}{$\spadesuit #1$}}}

\def\Node{{\hbox{Node}}}
\def\LL{{\hbox{LL}}}
\def\DSN{{\hbox{DSN}}}
\def\DCL{{\hbox{DCL}}}
\def\LS{{\hbox{LS}}}
\def\Ls{{\hbox{ls}}}

\def\FPV{{\hbox{FPV}}}
\def\Lfp{{\hbox{lfp}}}
\def\IsHeap{{\hbox{IsHeap}}}

\def\Equiv{\quad \equiv\quad }
\def\Null{{\hbox{null}}}
\def\Emp{{\hbox{emp}}}
\def\If{{\hbox{if\ }}}
\def\Then{{\hbox{\ then\ }}}
\def\Else{{\hbox{\ else\ }}}
\def\While{{\hbox{while\ }}}
\def\Do{{\hbox{\ do\ }}}
\def\Cons{{\hbox{cons}}}
\def\Dispose{{\hbox{dispose}}}
\def\Vars{{\hbox{Vars}}}
\def\Locs{{\hbox{Locs}}}
\def\States{{\hbox{States}}}
\def\Heaps{{\hbox{Heaps}}}
\def\FV{{\hbox{FV}}}
\def\True{{\hbox{true}}}
\def\False{{\hbox{false}}}
\def\Dom{{\hbox{Dom}}}
\def\Abort{{\hbox{abort}}}
\def\New{{\hbox{New}}}
\def\W{{\hbox{W}}}
\def\Pair{{\hbox{Pair}}}
\def\Lh{{\hbox{Lh}}}
\def\lh{{\hbox{lh}}}
\def\Elem{{\hbox{Elem}}}
\def\EEval{{\hbox{EEval}}}
\def\PEval{{\hbox{PEval}}}
\def\HEval{{\hbox{HEval}}}
\def\EVal{{\hbox{Eval}}}
\def\Domain{{\hbox{Domain}}}
\def\Exec{{\hbox{Exec}}}
\def\Store{{\hbox{Store}}}
\def\Heap{{\hbox{Heap}}}
\def\Storecode{{\hbox{Storecode}}}
\def\Heapcode{{\hbox{Heapcode}}}
\def\Lesslh{{\hbox{Lesslh}}}
\def\Addseq{{\hbox{Addseq}}}
\def\Separate{{\hbox{Separate}}}
\def\Result{{\hbox{Result}}}
\def\Lookup{{\hbox{Lookup}}}
\def\ChangeStore{{\hbox{ChangeStore}}}
\def\ChangeHeap{{\hbox{ChangeHeap}}}
\def\Wand{\mathbin{\hbox{\hbox{---}$*$}}}
\def\Eval#1{[\kern -1.5pt[{#1}]\kern -1.5pt]}
\def\Vec#1{{\bf #1}}

\def\Tilde{\widetilde}
\def\Break{\hfil\break\hbox{}}

%%% Local Variables:
%%% mode: latex
%%% TeX-master: "main.tex"
%%% End:

\title[Explicit Induction Is Not Equivalent to Cyclic Proofs]{Explicit Induction is Not Equivalent to Cyclic Proofs for Classical Logic with Inductive Definitions}

\author[Stefano Berardi]{Stefano Berardi\rsuper{a}} 
\author[Makoto Tatsuta]{Makoto Tatsuta\rsuper{b}}

\address{\lsuper{a}Universit\`{a} di Torino}
\address{\lsuper{b}National Institute of Informatics / Sokendai, Tokyo }

\begin{abstract}

A cyclic proof system, called CLKID-omega, gives us another way of
representing inductive definitions and efficient proof search.  The
2005 paper by Brotherston showed that the provability of
CLKID-omega
includes the provability of LKID, first order classical logic with inductive definitions in Martin-L\"of's style,
and conjectured the
equivalence.  The equivalence has been left an open
question since 2011.
This paper shows that CLKID-omega and LKID are indeed
not equivalent.  This paper considers a statement called 2-Hydra in
these two systems with the first-order language formed by 0, the successor,
the natural number predicate, and a binary predicate symbol used to
express 2-Hydra.  This paper shows that the 2-Hydra statement is
provable in CLKID-omega, but the statement is not provable in
LKID, by constructing some Henkin model where the statement is
false.

\end{abstract}

\maketitle

\Erase{}

\section{Introduction}
\label{section:introduction}

An inductive definition is a way to define a predicate by an expression which may contain the predicate itself.
The predicate is interpreted by the least fixed point of the defining equation on sets.
Inductive definitions are important in computer science, since they can define useful recursive data structures such as lists and trees.
Inductive definitions are important also in mathematical logic, since
they increase the proof theoretic strength.
Martin-L\"of's system of inductive definitions given in \cite{Martin-Lof-1971} is
one of the most popular system of inductive definitions.
This system has production rules for an inductive predicate,
and the production rules determine the introduction rule and the elimination rule for the predicate.

Brotherston \cite{Brotherston-phd} and Simpson
\cite{Brotherston11} proposed an alternative formalization
of inductive definitions, called a cyclic proof system.
A proof, called {\em a cyclic proof}, is defined by proof search,
going upwardly in a proof figure.
If we encounter the same sequent (called a bud)
as some sequent we already passed (called a companion), 
or we found anywhere else in the proof-tree,
we can stop.
The induction rule is replaced by a case rule, for this purpose.
The soundness is guaranteed by
some additional condition, called the {\em global trace condition},
which guarantees
that in any infinite path of the proof-tree there is some infinitely decreasing inductive definition.
In general, for proof search,
a cyclic proof system can find an induction formula in a more efficient way
than induction rules in Martin-L\"of's style,
since a cyclic proof system does not have to choose a fixed induction formula
in advance.
A cyclic proof system enables efficient implementation of
theorem provers with inductive definitions \cite{Brotherston05,Brotherston08}.
In particular, it works well for theorem provers of Separation Logic \cite{Brotherston11a,Brotherston12}.

Brotherston and Simpson \cite{Brotherston11}
investigated
the system $\LKID$ of inductive definitions in classical logic
for the first-order language,
and
the cyclic proof system $\CLKIDomega$ for the same language,
showed the provability of
$\CLKIDomega$ includes that of $\LKID$,
and conjectured the equivalence.
 Since then, the equivalence has been left an open question.
In $2017$, Simpson~\cite{Simpson-2017} proved a particular case
of the conjecture, for the theory of Peano Arithmetic.

 This paper (which is the journal version of \cite{Berardi-Makoto-FOSSACS2017}) shows $\CLKIDomega$ and $\LKID$ are indeed not equivalent.
To this aim, we will consider the first-order language $L$ (with equality) formed by $0$, the
successor $\s$, the natural number predicate $\Npredicate$, and a
binary predicate symbol~$p$.  We introduce a statement we call
2-Hydra, which is a miniature version of the Hydra problem considered
by Kirby and Paris \cite{Kirby-1982}: the proviso
``$2$'' means that we only have two ``heads''.
\Makoto{ $H$ as $\Delta,N(x),N(y) \prove p(x,y)$where $\Delta$ is the
set of assumptions for $H$, is not necessary here}
We show that the 2-Hydra
statement is provable in $\CLKIDomega$ with language $L$,
but the statement is not provable in $\LKID$
with language $L$. 2-Hydra is similar to the candidate for a counter-example proposed by Stratulat \cite{Stratulat-2016}. 

The unprovability is shown by constructing some Henkin
model $\setM$ of $\LKID$ where 2-Hydra is false. $2$-Hydra is true in all standard models of $\LKID$, but $\setM$ is a non-standard model, in which both the universe of $\setM$ and the interpretation of the predicate $\Npredicate$ are $\Nat + \Zeta$, where $\Nat$ is the set of natural numbers and $\Zeta$ is the set of integers. Predicates of $\setM$ are the equality relation and one ``partial bijection'', i.e., a one-to-one correspondence between subsets of the universe of $\setM$. The proof that $\setM$ is a Henkin model of $\LKID$ immediately follows from a quantifier elimination result, which holds for all sets of partial bijections which are closed under composition and inverse.

Our quantifier elimination result  is new, to our best knowledge, and it may be of some independent interest. However, our interest is not the quantifier elimination result \emph{per se}, but rather the identification of this result as a way of proving the unprovability of 2-Hydra in $\LKID$.

The model $\setM$ also shows a side result, that $\LKID$ is not conservative when we add
inductive predicates. Namely, it is not the case that for any language
$L$, the system of $\LKID$ with language $L$ and any additional
inductive predicate is conservative over the system of $\LKID$ with
$L$.

This is the plan of the paper. 
Section  \ref{section:induction} describes inductive definitions, standard and Henkin models.
Section  \ref{section:LKID} defines the first order system $\LKID$ for inductive definitions,
the 2-Hydra statement, and proves 2-Hydra under two additional assumptions: the $0$-axiom and the existence of an ordering $\le$.
Section  \ref{section:CLKID} defines the system $\CLKID$ for cyclic proofs, gives a cyclic proof for the 2-Hydra statement and describes the Brotherston-Simpson conjecture. 
Section  \ref{section:counter-model} defines the structure $\setM$
and the proof outline that $\setM$ is a counter model.
Section  \ref{section:partial-bijections} introduces a set of partial bijections in $\setM$.
Section  \ref{section:quantifier-elimination} proves a quantifier elimination theorem for any set of partial bijections closed under composition and inverse.
Section  \ref{section:main} disproves the Brotherston-Simpson conjecture, by proving that the 2-Hydra statement is not provable in $\LKID$. As a corollary, we have non-conservativity of $\LKID$ with additional inductive predicates.
We conclude in Section  \ref{section:conclusion}.

\section{Inductive Definitions, Standard Models and Henkin Models}
\label{section:induction}
In this section quickly recall the notion of first order inductive definition, standard model and Henkin model, taken from  \cite{Brotherston11}.  This introduction is only a sketch and we refer to \cite{Brotherston11} for motivations and examples. We fix a first order language $\Sigma$ with equality that includes inductive predicate symbols $P_1, \ldots,P_n$ of arities $k_1, \ldots, k_n$.

\begin{defi}[Productions of $\Sigma$]
\label{definition-inductive-predicate-symbols}
An inductive definition set $\Phi$ for $\Sigma$ is a finite set of productions. A production is a rule
$$
\infer
{P_i (\vec{t})}
{Q_1(\vec{u_1})  \ \ldots    \ Q_h( \vec{u_h} )   \ P_{j_1}  ( \vec{t_1} )   \  \ldots  \   P_{ j_m} (\vec{ t_m} )}
$$
whose premises are a finite sequence of atomic formulas, where 
$Q_1, \ldots ,Q_h$ are ordinary predicate symbols, $ j_1 , \ldots , j_m ,i \in \{1, \ldots ,n\}$, $P_1, \ldots, P_n$ are inductive predicate symbols,
and all vector of terms have the appropriate length to match the arities of the predicate symbols.
\end{defi}

An example. Let $\Sigma = \{0, \s\}$ be the language with $0$ and the successor. Then a set of productions $\Phi_{\Npredicate}$ describing the inductive predicate $\Npredicate$ for ``being a natural number'' is:  
\[
\begin{array}{cc}
\infer
{\Npredicate(0)}
{} 
\ \ \ 
& 
\ \ \ 
\infer
{\Npredicate(\s x)}
{\Npredicate( x)} 
\end{array}
\]
We call the pair $(\Sigma, \Phi)$ an inductive definition system. The language for $(\Sigma, \Phi)$ is the first order language consisting of all constants, functions and predicates of $\Sigma$. The standard interpretation for $(\Sigma, \Phi)$  is obtained by considering the smallest
prefixed point of a monotone operator $\phi_{\Phi}$ defined below. From now on, we denote the powerset of a set $X$ by $\setP(X)$. In the next definition we suppose that $\rho$ is a valuation from finitely many variables to the universe, and that $\rho$ is applied componentwise on a vector of terms.

\begin{defi}[Monotone Operator $\phi_{\Phi}$]
\label{definition:monotone-operator}
Let $\setM$  with domain $\universe{\setM}$ be a first-order structure for $\Sigma$, and for each $ i \in\{1, \ldots ,n\}$, let $k_i$ be the arity of
the inductive predicate symbol  $P_i$. 
\begin{enumerate}
\item
$\Phi_i =\{\phi  \in \Phi \mid  
\hbox{the conclusion of $\phi$ is $P_i$}
 \}$.
\item
Assume $\Phi_i$ has the form $\{\Phi_{i,r} \mid  1 \le r \le |\Phi_i|\}$. For each rule $\Phi_{i,r}$ of the form shown in
\ref{definition-inductive-predicate-symbols}, we define 
$\Phi _{i,r} :  
\setP(\universe{\setM}^{k_{j_1}} ) \times  \ldots \times  \setP(\universe{\setM}^{k_{j_n}} ) 
\rightarrow
 \setP(\universe{\setM}^{k_i} )$ 
by:
$$
\Phi_{i,r} (X_1  \ldots ,X_n)
= 
\{\rho(\vec{t}) \mid \rho \mbox{ a valuation},  \ 
\rho( \vec{t_1} ) \in X_{j_1} , \ \ldots , \ \rho(\vec{ t_m} ) \in X_{j_m}, 
$$
$$
Q^{\setM}_1 (\rho(\vec{u_1})), \ \ldots , \ Q^{\setM}_h (\rho( \vec{u_h} )) \}.
$$
\item
We define $\Phi_i$ for each
$i \in\{1, \ldots ,n\}$ with the same domain and codomain, by:
$$
\Phi_i(X_1, \ldots ,X_n)
=
\bigcup_{1 \le r \le |\Phi_i|}
\Phi_{i,r} (X_1, \ldots ,X_n).
$$
\item
We define $\phi _\Phi$, with domain and codomain
$\setP(\universe{\setM}^{k_1} ) \times  \ldots \times  \setP(\universe{\setM}^{k_n} ) $, by:
$\phi _\Phi(X_1, \ldots ,X_n)=(\Phi_1(X_1, \ldots ,X_n), \ldots ,\Phi_n(X_1, \ldots ,X_n))$.
\end{enumerate}
\end{defi}

We extend union and subset inclusion to the corresponding pointwise operations on n-tuples of sets: in this way $\dom(\phi_\Phi)$ becomes a complete lattice. A prefixed point of $\phi_\Phi$ is $\vec{X} \in \dom(\phi_\Phi)$ such that $\vec{X} \subseteq \phi_\Phi(\vec{X})$. The map $\phi_\Phi$ is monotone on a complete lattice. Thus,  $\phi_\Phi$ has a unique smallest prefixed point by the Tarski Fixed Point Theorem. We define the standard model for $(\Sigma, \Phi)$ from such a prefixed point.

\begin{defi}[Standard model]
A first-order structure $\setM$ with universe $\universe{\setM}$  is said to be a standard model for $(\Sigma, \Phi)$ if the vector $(P^{\setM}_1, \ldots, P^{\setM}_n)$ of interpretations of $P_1, \ldots, P_n$ in $\setM$ is the smallest prefixed point of $\phi_{\Phi}$.
\end{defi}

The Henkin class for $\setM$ is a family of subsets $\setH_k \subseteq \universe{\setM}^k$, indexed on $k \in \Nat$, including all graphs of predicates of $\setM$ and closed w.r.t. all first order connectives, as we make precise below.

\begin{defi}[Henkin class for a first order structure $\setM$]
\label{definition-henkin-family}
Le $\setM$  with domain $\universe{\setM}$ be a structure for $\Sigma$. Assume $k,h \in \Nat$, $\vec{x} = x_1, \ldots ,x_k$ are variables, $t_1, \ldots, t_h$ are terms and $\vec{d} = d_1, \ldots ,d_k \in \universe{\setM}$. A Henkin class for $\setM$ is a family of sets 
$\setH = \{ \setH_k  \mid  k  \in \Nat\}$ 
such that, for each $k  \in \Nat$: $\setH_k \subseteq \setP(\universe{\setM}^k)$;
\begin{enumerate}[label=$(H_{\arabic*})$]
\item % [$(H_1)$] 
$\{(d,d)  \mid d  \in \universe{\setM}\} \in \setH_2$;
\item % [$(H_2)$]
  if $Q \in \Sigma$ is any predicate symbol of arity $k$ then $Q_\setM \in  \setH_k$ ;
\item % [$(H_3)$]
  if $R \in  \setH_{k +1}$ and $d  \in \universe{\setM}$ then $\{(\vec{d})  \mid  (\vec{d},d) \in R\} \in  \setH_k $;
\item % [$(H_4)$]
  if $R \in  \setH_h$  and  $t_1[\vec{x}], \ldots ,t_h[\vec{x}]$ are terms then
$\{ (\vec{d})  \mid  ( t^\setM_1 [ \vec{d}], \ldots , t^\setM_h [ \vec{d}]) \in R \} \in H_k$;
\item % [$(H_5)$]
  if $R \in  \setH_k$  then $(\universe{\setM}^k \setminus R) \in  \setH_k $;
\item % [$(H_6)$] 
if $R_1,R_2  \in  \setH_k$  then $R_1 \cap R_2  \in  \setH_k $;
\item % [$(H_7)$] 
if $R \in  \setH_{k +1}$ then $\{(\vec{d}) \mid \exists d\in\universe{\setM}. (\vec{d},d) \in R\} \in  \setH_k $.
\end{enumerate}
\end{defi}

The smallest Henkin family $\setH_\calM$ for a structure $\setM$, and the only Henkin family we will consider later, is the set of definable sets in $\setM$.

\begin{defi}[The Henkin family $\setH_\setM$]
 \label{definition-definable-predicates}
Assume $\setM$ is a structure of language $\Sigma$. Let $k \in \Nat$. We write $\vec{u}, \vec{v}$ for two vectors of elements in $\universe{\setM}$ of the same length as the vectors of variables $\vec{x}, \vec{y}$. Then the family of sets $\setH_\setM = \{\setH_k \mid k \in \Nat\}$ is defined by:
\[
\setH_k = \{\ \
 \{ 
 \vec{u} \in \universe{\setM}^k  \mid 
\setM \models F[\vec{u},\vec{v}/\vec{x},\vec{y}] 
\} 
\ \ | \ \ 
(F\in L(\Sigma)) \ \wedge \
(\FV(F) \subseteq\vec{x}, \vec{y}) \  \wedge \
(\vec{v} \in \universe{\setM}) \} 
\]
\end{defi}

In a Henkin Model for $\Sigma$, instead of requiring that $(P^{\setM}_1, \ldots, P^{\setM}_n)$ is the smallest prefixed points of $\phi_\Phi$ w.r.t. all subsets of $\universe{\setM}$, we require that it is the smallest prefixed points w.r.t. all sets in some Henkin class $\setH_k$ for $\setM$.

\begin{defi}[Henkin model]
Let $\setM$ be a first-order structure for $(\Sigma, \Phi)$ and $\setH$ be a Henkin class for $\setM$. Then $(\setM,\setH)$ is a
Henkin model for $(\Sigma, \Phi)$ if $(P^\setM_1, \ldots, P^\setM_n)$ is the least prefixed point of $\phi_\Phi \lceil \setH$, where $\phi_\Phi \lceil \setH$ is $\phi_\Phi$ with each argument in $\setP(\universe{\setM}^k)$ being restricted to $\setH_k$.
We say $\calM$ is a Henkin model when $(\calM, \calH_\calM)$ is a Henkin model.
\end{defi}

In a Henkin model the interpretation of $P_1, \ldots, P_n$ may be larger than the smallest prefixed points of $\phi_\Phi$. Some Henkin models are not standard models, and this paper will discuss a Henkin model which is not a standard model.

\section{The system $\LKID$ for inductive definitions and the $2$-Hydra statement}
\label{section:LKID}
In this section we quickly introduce some definitions and results of \cite{Brotherston11}, in order to make the paper self-contained. We describe $\LKID(\Sigma, \Phi)$, formalizing the notion of inductive proof for a first order language $\Sigma$ with equality, and for the the set of productions $\Phi$ (see section \ref{section:induction}). We state that $\LKID$ is sound and complete with respect to Henkin models (again, see \cite{Brotherston11}). Then we formalize the $2$-Hydra statement, which is our work. In later sections we will prove that $2$-Hydra is false in some Henkin model, and we use $2$-Hydra to distinguish between provability in $\LKID$ and cyclic proofs.
\\

We write sequents of the form $\Gamma \vdash \Delta$ where $\Gamma, \Delta$ are finite sets of formulas. We write $\Gamma[\theta ]$ for the application of substitution $\theta$ to all formulas in $\Gamma$.
For first-order logic with equality, we use the (standard) sequent calculus rules, with contraction implicitly given. $\LKID(\Sigma, \Phi)$ has a rule for 
substitution, and rules for equality. There are logical rules and rules for inductive predicates.

Structural and logical rules of $\LKID$ are the following.
\\
\\
{\bf Structural rules:}

\[
\begin{array}{cc}

\infer[\Gamma \cap   \Delta \not =\emptyset    (\mbox{Axiom})]
{\Gamma \vdash \Delta}
{} 

& 

\infer[\Gamma \subseteq \Gamma' \ \ \  \Delta \subseteq \Delta'  \ \ \ (\mbox{Wk})]
{\Gamma' \vdash \Delta'}
{\Gamma \vdash \Delta  } 

\\
\\

\infer[(\mbox{Cut})]
{\Gamma \vdash \Delta}
{\Gamma \vdash \Delta, F & F,\Gamma \vdash \Delta}

 &
 \infer[(\mbox{Subst})]
 {\Gamma[\theta] \vdash \Delta[\theta]}
 {\Gamma \vdash \Delta} 
 \\
\end{array}
\]
\\  %16:21 01/11/2018
\\
{\bf Logical rules:}

\[
\begin{array}{ccc}

\infer[(\neg L)]{\Gamma, \neg F \vdash \Delta}{\Gamma \vdash F, \Delta} 
&  
\infer[(\neg R)]{\Gamma \vdash \neg F, \Delta}{\Gamma,F \vdash \Delta} 

\\
\\

\infer[(\vee L)]{\Gamma, F \vee G \vdash \Delta}{\Gamma, F \vdash \Delta & \Gamma, G \vdash \Delta} 
&  
\infer[(\vee R)]{\Gamma \vdash F \vee G, \Delta}{\Gamma \vdash F, G, \Delta}

\\
\\

\infer[(\wedge L)]{\Gamma, F \wedge G \vdash \Delta}{\Gamma, F, G \vdash \Delta} 
&  
\infer[(\wedge R)]{\Gamma \vdash F \wedge G, \Delta}{\Gamma \vdash F, \Delta & \Gamma \vdash G, \Delta}  

\\
\\

\infer[(\rightarrow L)]{\Gamma, F \rightarrow G \vdash \Delta}{\Gamma \vdash F, \Delta & \Gamma, G \vdash \Delta} 
&  
\infer[(\rightarrow R)]{\Gamma \vdash F \rightarrow G, \Delta}{\Gamma, F \vdash G, \Delta}

\\
\\

\infer[(x \not \in \FV(\Gamma,\Delta) \ \ \  (\exists L)]
{\Gamma, \exists x. F \vdash \Delta}
{\Gamma, F \vdash \Delta } 

&  
\infer[(\exists R)]
{\Gamma \vdash \exists x.F, \Delta}
{\Gamma \vdash F[t/x], \Delta}

\\
\\

\infer[(\forall L)]
{\Gamma, \forall x. F \vdash \Delta}
{\Gamma, F[t/x] \vdash \Delta } 

&  
\infer[(x \not \in \FV(\Gamma,\Delta)  \ \ \  (\forall R)]
{\Gamma \vdash \forall x.F, \Delta}
{\Gamma \vdash F, \Delta}

\\
\\

\infer[(=\!\!\!L)]{\Gamma[t/x,u/y], t=u\vdash \Delta[t/x,u/y]}{\Gamma[u/x,t/y]\vdash \Delta[u/x,t/y]} 
&  
\infer[(=\!\!\!R)]{\Gamma \vdash t=t, \Delta}{}  

\end{array}
\]
\\
\\

We define left- and right-introduction rules for induction. For each production in $\Phi$ of the form:
$$
\infer
{P_i (\vec{t}[\vec{x}])}
{Q_1(\vec{u_1}[\vec{x}])  \ldots  Q_h( \vec{u_h}[\vec{x}] ) P_{j_1}  ( \vec{t_1}[\vec{x}] )  \ldots  P_{j_m} (\vec{ t_m} [\vec{x}])}
$$
we include the following right-introduction rule for  $P_i$ in $\LKID(\Sigma, \Phi)$:
$$
\infer
{\Gamma \vdash \Delta, P_i (\vec{t}[\vec{u}])}
{\Gamma \vdash \Delta, Q_1(\vec{u_1}[\vec{u}])  \ldots  \Gamma \vdash \Delta,Q_h( \vec{u_h}[\vec{u}] )  &  \Gamma \vdash \Delta, P_{j_1}  ( \vec{t_1}[\vec{u}] )  \ldots  \Gamma \vdash \Delta, P_{j_m} (\vec{ t_m}[\vec{u}] )}
$$
We assume that $\vec{u}$ is a vector of terms of the same length as $\vec{x}$.

We express left-introduction rules for inductive predicates in the form of induction rules for mutually depending predicates. We define mutual dependency first.

\begin{defi}[Mutual dependency \cite{Brotherston11}]
Let $P_i, P_j$ be inductive predicate symbols of $\Sigma$. 
\begin{enumerate}
\item
$P_j$ is a premise of $P_i$ if $P_i$  occurs in the conclusion of some production in $\Phi$, and $P_j$ occurs among the premises
of that production. 
\item
$P_i$  and $P_j$ are mutually dependent if there is a chain for the \emph{``premise relation''} from $P_i$ to $P_j$, and conversely.
\end{enumerate}
\end{defi}

In order to define the left-introduction rule for any inductive predicate $P_j$, we first
associate with every inductive predicate  $P_i$  a tuple $\vec{z}_i$ of $k_i$ distinct variables (called induction variables), where $k_i$ is the arity of  $P_i$, and a formula (called an induction hypothesis) $F_i$, possibly containing (some of) the induction variables $\vec{z}_i$. We define a formula $G_i$ for each $i \in \{1, \ldots ,n\}$ by:
$G_i= F_i$ if  $P_i$  and $P_j$ are mutually dependent
and  $G_i= P_i (\vec{z}_i)$ otherwise.
We write $G_i\vec{t}$ for $G_i[\vec{t}/\vec{z}_i]$, and the same for $F_i$. Then the induction rule for $P_j$ has the following form:
$$
\infer
{\Gamma,P_j\vec{u} \vdash \Delta}
{\mbox{minor premises} & \Gamma,F_j\vec{u} \vdash \Delta}
$$
The premise $\Gamma,F_j\vec{u} \vdash \Delta$ is called the major premise of the rule, and for each production
of $\Phi$ having in its conclusion a predicate $P_i$  that is mutually dependent with $P_j$, say:
$$
\infer
{P_i (\vec{t}[\vec{x}])}
{Q_1(\vec{u_1}[\vec{x}])  \  \ldots  Q_h( \vec{u_h}[\vec{x}] )  \  P_{j_1}  ( \vec{t_1}[\vec{x}] )  \   \ldots    \  P_{j_m} (\vec{ t_m} [\vec{x}])}
$$
there is a corresponding minor premise:
$$
\Gamma,Q_1\vec{u_1}[\vec{y}],  \ldots  ,Q_h \vec{u_h} [\vec{y}],G_{j_1}   t_1 [\vec{y}],  \ldots  ,G_{ j_m}\vec{t_m} [\vec{y}] \vdash F_i\vec{t}[\vec{y}],\Delta
$$
where $\vec{y}$ is a vector of the same length as $\vec x$
for fresh variables.

An alternative formalization of induction is the \emph{induction schema}. 

\begin{defi}[Induction schema]
\label{definition-induction-schema}
The induction schema  is the following set of axioms: 
$$
\mbox{(universal closures of minor premises) } 
\rightarrow \forall \vec{y}.(P_i\vec{t}(\vec{y}) \rightarrow F_i\vec{t}(\vec{y})), \ \  \ 
\mbox{ for } i\in\{1, \ldots, n\}
$$ 
$\setM$ is defined to \emph{satisfy the induction schema} if and only if all formulas of the induction schema are true in $\setM$.
\end{defi}

The axioms in the induction schema derive all instances of the induction rule and conversely. 
By definition unfolding we have the following. If a structure $\calM$
has $0,s,N$ and the inductive predicate symbol is only $N$,
then $(\setM,\setH_\setM)$ is
a Henkin model if and only if: if $F[0/x]$ is true, and for all closed
terms $t$ if $F[t/x]$ is true then $F[\s t/x]$ is true, then we have
$F[u/x]$ true, for all closed terms $u$. Thus, such a structure $\setM$ is a Henkin
model if and only if $\setM$ satisfies the induction
schema. The same remark applies to all inductive predicates.

We write $A_1, \ldots, A_n
\implies B$ for $A_1 \wedge \ldots \wedge A_n \implies C$ and $\forall
x_1, \ldots, x_n \in \Npredicate. \ A$ for $\forall x_1. \ \ldots. \ \forall
x_n. (\Npredicate(x_1) \wedge \ldots \wedge \Npredicate(x_n) \implies
A)$. 
We abbreviate $1$ and $2$ for ${\s} 0$ and $\s \s 0$ respectively.

{\bf The case of the predicate $\Npredicate$}. The induction rule for the `natural number' predicate $\Npredicate$ 
(section  \ref{section:induction}) is:
$$
\infer[(\mbox{Ind } N)]{\Gamma,Nt\vdash \Delta}{\Gamma \vdash F0,\Delta &  \Gamma,Fx\vdash Fsx,\Delta &  \Gamma,Ft\vdash\Delta}
$$

where $x$ is fresh and $F$ is the induction formula associated with the predicate $\Npredicate$. The induction schema for $\Npredicate$ is the set of axioms 
$$
F0, (\forall x.Fx\rightarrow Fsx) \ \rightarrow \forall x.(Nx \rightarrow Fx)
$$ 
for any formula $F$. $\setM$ \emph{satisfies the induction schema} for $\Npredicate$ if and only if $\setM \models F0, (\forall x.Fx\rightarrow Fsx) \ \rightarrow \forall x.(Nx \rightarrow Fx)$.

\begin{defi}[Validity and Henkin Validity]
A sequent $\Gamma \vdash \Delta$ is said to be valid if it is true in all standard models. Let $(\setM,\setH)$ be a Henkin model for $\LKID(\Sigma,\Phi)$. A sequent $\Gamma \vdash \Delta$ is said to be true in $(\setM,\setH)$ if, for all valuations $\rho$ for $\setM$, whenever $\setM \models_\rho J$ for all $J  \in \Gamma$  then $\setM \models_\rho K$ for some $K  \in \Delta$. A sequent is said to be Henkin valid if it is true in all Henkin models. 
\end{defi}

A derivation tree is a tree of sequents in which each sequent is obtained as the
conclusion of an inference rule with its children as premises. A proof in $\LKID$ is a finite
derivation tree all of whose branches end in an axiom. The basic result about provability is:

\begin{thm} [Henkin soundness and completeness of $\LKID$ \cite{Brotherston11}]
$\Gamma \vdash \Delta$ is a sequent provable in $\LKID$ if and only if $\Gamma \vdash \Delta$ is Henkin valid.
\end{thm}

We refer to \cite{Brotherston11} for a proof. Completeness does not hold for (standard) validity: there are valid sequents with no proof in $\LKID$. One example is the sequent $\vdash H$, where $H$ is the $2$-Hydra statement, defined below.

\subsection{The Hydra Problem}
The Hydra of Lerna was a mythological monster, popping two smaller heads whenever you cut one. It was a swamp creature (its name means ``water'') and possibly was the swamp itself, whose heads are the swamp plants, with two smaller plants growing whenever you cut one. The original Hydra was defeated by fire, preventing heads from growing again. In the mathematical problem of Hydra, we ask whether it is possible to destroy an Hydra just by cutting heads.

Kirby and Paris \cite{Kirby-1982} formulated the Hydra problem as a statement for mathematical trees. We are interested about making Hydra a problem for natural numbers, representing the length of a head, and restricting to the case when the number of heads is always $2$. We call our statement $2$-Hydra. It is a miniature version of the Kirby-Paris statement. $2$-Hydra will give a counterexample to the Brotherston-Simpson conjecture.

\subsection{The $2$-Hydra Statement}

In this subsection we give the 2-Hydra statement,
which is a formula saying that any 2-Hydra eventually loses its two heads.

Let $\Sigma_\Npredicate$ be the signature $\{0, \s, p, N \}$,
which are zero, the successor, an ordinary binary predicate symbol
$p$, and an inductive predicate $\Npredicate$ for natural numbers.
The logical system $\LKID(\Sigma_\Npredicate,\Phi_\Npredicate)$ is
defined as the system $\LKID$ with the signature $\Sigma_N$ and the
production rules $\Phi_N$ for $\Npredicate$ (see section
\ref{section:induction}). We define the $(0,\s )$-axioms in this
language as the axioms ``$0$ is not successor'' or $\forall x \in
\Npredicate.  \ \s x \not = 0$, and ``successor is injective'', or
$\forall x, y \in \Npredicate. \ \s x= \s y \implies x = y$.

We consider a formal statement $H$ for $2$-Hydra. $H$ says that the number of heads
is always $2$, and we can win a game having the following rules:
\begin{enumerate}
\item When both heads have positive length, we cut them off completely.
Then the first head shrinks to become $1$ unit shorter than the previous head and the second head shrinks to become $2$ units shorter than the previous head, if these shorter lengths exist.
Otherwise we win.
\item
When there is a unique head of positive length, we cut it off completely.
Then the first head shrinks to become $1$ unit shorter than the original head of
positive length and the second head shrinks to become $2$ units shorter than the original head of positive length, if these shorter lengths exist.
Otherwise we win.
\end{enumerate}
We express $H$ by saying that some set of transformations eventually reaches a winning condition. The winning condition is the union of the winning conditions for the points $1$ and $2$ above. Let $n, m \in \Nat$. The the winning conditions and the transformations  are:  
\begin{enumerate}
\item
we win if we reach the cases: $(0,0)$, $(1,0)$ and $(x,1)$ for any $x \in \Nat$. 
\item
if $n \ge 1$ and $m \ge 2$ then
$(n,m) \mapsto (n-1,m-2)$; 
\item
if $m \ge 2$ then $(0,m) \mapsto (m-1,m-2)$;
\item
if $n \ge 2$ then $(n,0) \mapsto (n-1,n-2)$; 
\end{enumerate}
The four cases listed above are pairwise disjoint and cover all $(n, m) \in \Nat^2$. For instance, case $4$ is disjoint from cases $1, 2, 3$. When we win, no transformation applies. Indeed, no transformation applies from $(0,0)$ and $(1,0)$, because if $m=0$ we require $n \ge 2$. No transformation applies when $m=1$, because transformations $1$ and $2$ require $m \ge 2$, and transformation $3$ requires $m=0$. We define $H$ by a formula in  the language $\Sigma_\Npredicate$. 

\begin{defi}[2-Hydra Statement $H$]
\label{definition:H}
We define
$H = (H_a,H_b,H_c,H_d \implies \forall x,y \in \Npredicate. \ p(x,y))$, where $H_a,H_b,H_c,H_d$ are:
\begin{enumerate}[label=$(H_{\alph*})$]
\item % [$(H_a)$]
$\forall x \in \Npredicate. \  p(0,0) \wedge p(1,0) \wedge p(x,1)$,
\item % [$(H_b)$]
$\forall x,y \in \Npredicate. \  p(x,y) \implies p(\s x, \s \s y)$,
\item % [$(H_c)$]
$\forall y \in \Npredicate. \  p(\s y,y) \implies p(0,\s \s y)$,
\item % [$(H_d)$]
$\forall x \in \Npredicate. \  p(\s x,x) \implies p(\s \s x,0)$.
\end{enumerate}
\end{defi}

For a closed term of $\{0,\s\}$,
its length is defined as the number of symbols $\s$ in it. Assume $n, m$ are closed terms of $\{0,\s\}$. Then $p(n,m)$ means that we win 
for the $2$-Hydra game beginning with the first head being of length $n$
and the second head being of length $m$.
For all closed terms $n,m$ of $\{0,\s\}$, there is a formula among $H_a,H_b,H_c,H_d$ having some instance inferring $p(n,m)$. The formula is unique if we assume the standard $(0,\s)$-axioms. 
$H_a$ says that $p(0,0)$, $p(1,0)$ and $p(n,1)$ for any closed term $n$ are true, and expresses the winning condition of the game. Each instance of $H_b,H_c,H_d$ is some implication 
$p(n',m') \rightarrow p(n,m)$ such that the maximum length of $n',m'$ is smaller than 
the maximum length of $n,m$. Thus, for all closed terms $n,m$ of $\{0,\s\}$, 
$p(n,m)$ is true in all standard models of $\LKID(\Sigma_N,\Phi_N)$: it is shown by induction on the maximum length of $n,m$. An example: we derive $p(1,4) $ by $H_b$ and $p(0,2)$, the latter by $H_c$ and $p(1,0)$, the latter by $H_a$. In a standard model, the interpretation of $\Npredicate$ is the set of interpretations of closed terms of $\{0,\s\}$: as a consequence, $2$-Hydra is is true in all standard models of $(\Sigma_\Npredicate, \Phi_\Npredicate)$.

However, we will prove that $\LKID(\Sigma_\Npredicate, \Phi_\Npredicate) + (0,\s)$-axioms does not prove $2$-Hydra. Remark that the $(0,\s)$-axioms define a proper extension of $\LKID(\Sigma_\Npredicate, \Phi_\Npredicate)$. These axioms cannot be proved in $\LKID(\Sigma_\Npredicate,
\Phi_\Npredicate)$, because each of them fails in the following models.
\begin{enumerate}
\item 
the model with domain
$\universe{\setM} = \Npredicate_\setM = \{0\}$, $\s 0=0$; 
\item
the model with domain $\universe{\setM} = \Npredicate_\setM = \{0, \s 0\}$, $0 \not = \s
0$ and $\s \s 0 = \s 0$. 
\end{enumerate}
Compared with Peano Arithmetic $\PA$, in
$\LKID(\Sigma_\Npredicate, \Phi_\Npredicate) + (0,\s)$-axioms we do
not have a sum or a product on $\Npredicate$,
and we do not have inductive predicate symbols for addition, multiplication, 
or order.

\subsection{$2$-Hydra is provable under additional assumptions}
As an example of a formal proof in $\LKID$, we prove $2$-Hydra under two additional assumptions: the inductive predicate~$\le$ and the $0$-axiom, which will be defined.

The inductive predicate $\le$ is defined from the following production rules:
\[
\infer{x \le x}{} \qquad \infer{x \le \s y}{x \le y}
\]
We call the set of these production rules $\Phi_{\le}$.
The 0-axiom is: $\forall x \in \Npredicate. \  \s x \not = 0$.
In $\LKID(\Sigma_N+\{\le\},\Phi_N+\Phi_{\le})$, we can show
any number $\le 0$ is only 0.

\begin{lem}
\label{lemma:less-than-zero}
$\zeroaxiom, Nx,Ny, x \le y \prove y=0 \implies x=0$
\end{lem}

\begin{proof}[Proof]

The proof is by induction on the definition of $x \le y$. If $y$ is
$x$ then $x=0 \implies x=0$, if $y$ is $\s(z)$ and the property holds
for $x, z$ then we trivially have $\s(z)=0 \implies x=0$ by
$\zeroaxiom$.
\end{proof}

The next theorem shows 2-Hydra is provable in $\LKID$ with $\le$.
\begin{thm}\label{theorem:hydra-order}
$\zeroaxiom \prove H$ is provable in $\LKID(\Sigma_N+\{\le\},\Phi_N+\Phi_{\le})$.
\end{thm}

\begin{proof}[Proof]  %[Proof of Theorem \ref{theorem:hydra-order}]
Let $\hat H = H_a,H_b,H_c,H_d$ be the list of 2-Hydra axioms in Def. \ref{definition:H}.
We will prove the equivalent sequent $\hat H, Nx,Ny \prove p(x,y)$.
We will first show $\forall n.  \ (n \ge x \land n \ge y \imp p(x,y))$
by induction on $n$.

\begin{itemize}
\item
Case 1: $n=0$. Then $x=y=0$ by Lemma \ref{lemma:less-than-zero}, therefore $p(x,y)$ by $H_a$.
\item
Case 2: $n=\s n'$.

\begin{itemize}
\item
  Sub-case 2.1: $y=0$.
  \begin{itemize}
  \item Sub-sub-case 2.1.1. $x=0$ or $x=\s 0$. By $H_a$.
  \item Sub-sub-case  2.1.2. $x=\s\s x''$. Then $p(\s x'',x'')$ by induction hypothesis, hence $p(x,0)$ by $H_d$.
  \end{itemize}
  
\item
  Sub-case 2.2. $y=\s 0$. By $H_a$.

\item 
  Sub-case 2.3. $y=\s\s y''$.
  \begin{itemize}
    \item Sub-sub-case  2.3.1. $x=0$. Then $p(\s y'',y'')$ by I.H., hence
      $p(0,y)$ by $H_c$.
    \item Sub-sub-case  2.3.2. $x=\s x'$. Then $p(x',y'')$ by I.H., therefore $p(x,y)$ by $H_b$.
    \end{itemize}
  \end{itemize} 
\end{itemize}

\noindent By principal induction on $x$ and secondary induction on $y$ we prove that $\exists n. (n \ge x) \wedge (n \ge y)$. From this statement and the previous one we conclude our claim.
\end{proof}

\section{The system $\CLKID$ and the Brotherston-Simpson conjecture}
\label{section:CLKID}
In this section we introduce more definitions and results of \cite{Brotherston11}, again in order to make the paper self-contained. We define an infinitary version of $\LKID$ called $\LKID^\omega$, then a subsystem $\CLKID$ of the latter called the system of cyclic proofs in \cite{Brotherston11}. We give a cyclic proof of $2$-Hydra, which is ours, and eventually we state the Brotherston-Simpson conjecture. 

The proof rules of the infinitary system $\LKID^\omega$  are the rules of $\LKID$, except the induction rules for each inductive predicate. The induction rule $(\mbox{Ind }P_i)$ of $\LKID$ is replaced by the case-split rule:
$$
\infer[(\mbox{Case  }P_i )]{\Gamma,P_i\vec{u}\vdash \Delta}{\mbox{case distinctions}}
$$
with case distinctions defined as follows. 
For each production having predicate $P_i$  in its conclusion:
$$
\infer{P_i\vec{t}[\vec{\vec{x} }]}
{Q_1 \vec{u} _1[\vec{x} ] \ \ldots \ Q_h \vec{u} _h [\vec{x} ] P_{j_1}   \vec{t} _1 [\vec{x} ] \ 
\ldots  \ P_{j_m}  \vec{t} _m [\vec{x} ]}
$$

there is a case distinction
$$
\Gamma, \vec{u}  = \vec{t} [\vec{y} ],Q_1\vec u_1[\vec{y} ],  \ldots  ,Q_h \vec{u} _h [\vec{y} ], 
P_{j_1}\vec{t} _1 [\vec{y} ],  \ldots  , P_{j_m} \vec{t} _m [\vec{y} ] \vdash \Delta
$$
where $\vec{y}$ is a vector of distinct variables of the same length as $\vec{x}$, and $\vec{y} \cap V = \emptyset$ for $V = FV(\Gamma \cup  \Delta \cup  \{P_i\vec{u} \mid  i=1, \ldots, n\}) $.

The formulas $P_{j_1}\vec{t} _1 [\vec{y} ],  \ldots  , P_{j_m} \vec{t} _m [\vec{y} ]$ occurring in a case distinction are said to be case-descendants of the principal formula $P_i\vec{u}$.

The case-split rule for N is
$$
\infer[x \not \in \FV(\Gamma \cup   \Delta \cup \{Nt\}) \  (\mbox{Case }N)]
{\Gamma,Nt \vdash \Delta}
{\Gamma,t=0\vdash \Delta & \Gamma,t=sx,Nx\vdash\Delta}
 $$

The formula $Nx$ occurring in the right hand premise is the only case-descendant of the formula $Nt$
occurring in the conclusion.

The system $\LKID^\omega$  is based upon infinite derivation trees. We
distinguish between `leaves' and `buds' in derivation trees. By a leaf we mean an axiom, i.e. the
conclusion of a 0-premise inference rule. By a bud we mean a sequent occurrence in the tree that
is not the conclusion of a proof rule. 

\begin{defi}[$\LKID^\omega$  Pre-proof]
An $\LKID^\omega$  pre-proof of a sequent $\Gamma \vdash \Delta$ is a (possibly infinite) derivation tree $\Pi$, constructed
according to the proof rules of $\LKID^\omega$, such that $\Gamma \vdash \Delta$ is the root of $\Pi$ and $\Pi$ has no buds.
\end{defi}

$\LKID^\omega$  pre-proofs are not sound in general: there are 
pre-proofs of any invalid sequent.
The global trace condition is a condition on pre-proofs which ensures their soundness.

A (finite or infinite) path $\pi$ in a derivation tree $\Pi$ is a sequence $\pi = (S_i)_{0\le i<\alpha}$,
for some $\alpha  \in N\cup   \{ \infty \}$, of sequent occurrences in the tree such that $S_{i+1}$ is a child of $S_i$ for all $i+1<\alpha$.

\begin{defi}[Trace]
Let $\Pi$ be an $\LKID^\omega$  pre-proof and let $\pi = (\Gamma_i \vdash \Delta_i)_{i\ge 0}$ be an infinite path in $\Pi$. 
A trace following $\pi$ is
a sequence $\tau = (\tau_{i})_{i\in\Nat}$ such that, for all $i \in\Nat$:
\begin{enumerate}
\item  $\tau_{i}=P_{j_i} \vec{t}_i  \in \Gamma_i$, where $j_i  \in\{1, \ldots ,n\}$;
\item  if $\Gamma_i \vdash \Delta_i$ is the conclusion of $(\mbox{Subst})$ then $\tau_{i}=\tau_{i+1}[\theta ]$, where $\theta$  is the substitution associated with the instance of $(\mbox{Subst})$;
\item  if $\Gamma_i \vdash \Delta_i$ is the conclusion of $(=\!\!\!L)$ with principal formula $t=u$ then there is a formula $F$ and variables $x,y$ such that $\tau_{i}=F[t/x,u/y]$ and $\tau_{i+1}=F[u/x,t/y]$;
\item  if $\Gamma_i \vdash \Delta_i$ is the conclusion of a case-split rule then either \emph{(a)} $\tau_{i+1}=\tau_{i}$ or \emph{(b)} $\tau_{i}$ is the principal formula of the rule instance and $\tau_{i+1}$ is a case-descendant of $\tau_{i}$. In the latter case, $i$ is said to be a progress point of the trace;
\item  if $\Gamma_i \vdash \Delta_i$ is the conclusion of any other rule then $\tau_{i+1}=\tau_{i}$.
\end{enumerate}
\end{defi}

An infinitely progressing trace is a trace having infinitely many progress points.

\begin{defi}[$\LKID^\omega$ proof]
An $\LKID^\omega$  pre-proof $\Pi$ is defined to be an $\LKID^\omega$  proof if it satisfies the following global trace condition: for every infinite path $\pi = (\Gamma_i \vdash \Delta_i)_{i\ge 0}$ in $\Pi$, there is an infinitely progressing trace following some tail of
the path, $\pi' = (\Gamma_i \vdash \Delta_i)_{i\ge k}$, for some $k\ge 0$.
\end{defi}

Cyclic proofs are a subsystem $\CLKID$ of $\LKID^\omega$, defined by
restricting $\LKID^\omega$  to proofs given by regular trees, i.e. those (possibly infinite) trees with only
finitely many distinct subtrees. 

Proofs of $\CLKID$ are called cyclic proofs and are represented as finite graphs.

\begin{defi}[Companion]
Let $B$ be a bud of a finite derivation tree $\Pi$. A node $C$ in $\Pi$ which is conclusion of some rule is said to be a companion for $B$ if $C$ and $B$ are the same sequent.
\end{defi}

\begin{defi}[Cyclic pre-proof]
A C$\LKID^\omega$  pre-proof $\Pi$ of $\Gamma \vdash \Delta$ is a pair $(\Pi,R)$, where $\Pi$ is a finite derivation tree constructed according to the rules of $\LKID^\omega$  and whose root is $\Gamma \vdash \Delta$, and $R$ is a function assigning a companion to every bud node in $\Pi$.
\end{defi}

By unfolding a cyclic pre-proof to its associated (possibly infinite) tree, cyclic
pre-proofs generate exactly the class of $\LKID^\omega$  pre-proofs given by the regular derivation trees.

\begin{defi}[Cyclic proof]
A $\CLKID$  proof is defined as a $\CLKID$  pre-proof such that its unfolding satisfies the global trace condition.
\end{defi}

\subsection{Proof of 2-Hydra Statement in Cyclic-Proof System}
\label{section-cyclic-proof}
The logical systems  $\LKID(\Sigma_N,\Phi_N)$ and $\CLKID(\Sigma_N,\Phi_N)$ are
the systems $\LKID$ and $\CLKID$ with the signature $\Sigma_N$ and 
the set of production rules $\Phi_N$.
In this subsection we give an example of a cyclic proof: a cyclic proof of the $2$-Hydra statement in $\CLKID(\Sigma_N,\Phi_N)$. 

\begin{thm}\label{theorem:cyclic}
The 2-Hydra statement $H$ is provable in $\CLKID(\Sigma_\Npredicate, \Phi_\Npredicate)$.
\end{thm}

\begin{proof}
Let the 2-Hydra axioms $\hat H$ be $H_a,H_b,H_c,H_d$ as in Definition \ref{definition:H}.

For simplicity, we will write $p(x,y)$ as $pxy$,
and we will write the use of $2$-Hydra axioms 
by omitting (Cut), 
$(\imp\ R)$,
$(\forall\ R)$,
(Axiom),
as in the following example.
\[
	\infer{\hat H,N{\s} y'',Ny'' \prove p0\s\s y''}{
	\hat H,N{\s} y'',Ny'' \prove p\s y''y''
	}
\]
Rule (=L) is left-introduction of equality: from $\Gamma[a,b] \vdash \Delta[a,b]$ prove $a=b,\Gamma[b,a] \vdash \Delta[b,a]$. We will  write a combination of (Case) and (=L) 
as one rule in the following example.
\[
\infer[Nx]{\hat H,Nx \prove px0}{
	\hat H, N0 \prove p00
	&
	\hat H,Nx' \prove p\s x'0
}
\]

For saving space, we omit writing $\hat H$ in every sequent in the next proof figure.
For example, $Nx,Ny \prove pxy$ actually denotes
$\hat H,Nx,Ny \prove pxy$.

We define a cyclic proof $\Pi$ of $Nx,Ny \prove pxy$, where the mark \emph{(a)} denotes the bud-companion relation (there are three buds, the only companion is the root). $\Pi$ is:
\\
\[
\infer[Ny]
 {\emph{\bf (a)} Nx,Ny \prove pxy}
 {
  \infer{Nx \prove px0}{\deduce{\Pi_1}{\ldots \emph{\bf { (a)}} \ldots \ } }
  \ \ \ 
  & 
  \ \ \ 
  \infer{Nx,Ny' \prove px\s y'}{\deduce{\Pi_2}{\ldots \emph{\bf { (a)}} \ldots \emph{\bf { (a)}} \ldots \  } }
 }
\]
\\
where the left sub-proof $\Pi_1$ is:
\\
\[
\infer[Nx]
{Nx \prove px0}
{
	\infer{N0 \prove p00}{} \ \ \
	&
	\kern -0.5cm 
	\ \ \
	\infer[Nx']{Nx' \prove p\s x'0}
	{
		\infer{N0 \prove p10}{}
		&
		\infer
		{N{\s} x'',Nx'' \prove p\s\s x''0}
		{\infer{N{\s} x'',Nx'' \prove p\s x''x''} {\emph{\bf  (a)} Nx,Ny \prove pxy}}
	}
}
\]
\\
and the right sub-proof $\Pi_2$ is:
\\
\[
{
\infer[Nx]
{Nx,Ny' \prove px\s y'}
{
   \infer{N0,Nx \prove px1}{}
	&
	\infer[Nx]
	 {N{\s} y'',Nx,Ny'' \prove px\s\s y''}
	 {
		\infer
		{N{\s} y'',Ny'' \prove p0\s\s y''}
		{
		  \infer
		  {N{\s} y'',Ny'' \prove p\s y''y''}
		  {\emph{\bf  (a)} Nx,Ny \prove pxy}
		}
		&
		\infer
		{N{\s} y'',Nx',Ny'' \prove p\s x'\s\s y''}
		{
		  \infer
		  {Nx',Ny'' \prove px'y''}
		  {\emph{\bf  (a)} Nx,Ny \prove pxy}
		}
	}
}
}
\]

\paragraph{\emph{$\Pi$ is a cyclic proof.}}
We only have to check: the global trace condition holds for any infinite path $\pi$ in the cyclic proof $\Pi$ above. We can explicitly describe an infinite trace in $\pi$, as follows. We have three possible choices for constructing the infinite path $\pi$ in the proof: taking the bud in the left, middle, or right of the proof. For a given bud and $z_1,z_2 \in \{x,y\}$, we write $z_1 \leadsto z_2$ for a \emph{progressing}
trace from ${\Npredicate}z_1$ in the companion to ${\Npredicate}z_2$ in the bud. We write $z_1
\leadsto z_2,z_3$ for $z_1 \leadsto z_2$ and $z_1 \leadsto z_3$.
For the left bud, there are $x \leadsto x,y$.
For the middle bud, there are $y \leadsto x,y$.
For the right bud, there are both $x \leadsto x$ and $y \leadsto y$. We argue by cases.
\begin{enumerate}
\item
Assume that from some point on the left bud does not appear in a path. Then from this point there is an infinitely progressing trace $y \leadsto y \leadsto y \leadsto \ldots$.
\item
Assume that the middle bud from some point on does not appear in a path. Then from this point there is an infinitely progressing trace $x \leadsto x \leadsto x \leadsto \ldots$. 
\item
Assume that the left and middle buds appear infinitely many times in a path. Start from $x$ and the left bud if the left bud comes first, and from $y$ and the middle bud if the middle bud comes first, then repeat infinitely one of following operations, according to the current bud.
Take $x \leadsto x$ for all left buds, except for the last left bud before
the middle bud comes. Take $x \leadsto y$ for this bud.
Take $y \leadsto y$ for all middle buds, except for the last middle bud before
the left bud comes. Take $y \leadsto x$ for this bud.

In both cases, take $x \leadsto x$ or $y \leadsto y$ for the right bud, depending if the previous trace was ending in $x$ or in $y$. Also in this case we defined an infinitely progressing trace starting from some tail of the path, passing infinitely many times though $x$ and though $y$.
\end{enumerate}
Hence the global trace condition holds.
\end{proof}

\subsection{The Brotherston-Simpson Conjecture}

$\LKID$ has been often used for formalizing inductive definitions,
while $\CLKID$ is another way for formalizing the same inductive definitions,
and moreover $\CLKID$ is more suitable for proof search.
This raises the question of the relationship between $\LKID$ and
cyclic proofs: Brotherston and Simpson conjectured the equality for
each inductive definition. The left-to-right inclusion is proved in
\cite{Brotherston-phd}, Lemma 7.3.1 and in
\cite{Brotherston11}, Thm. 7.6. 
The Brotherston-Simpson conjecture (the conjecture 7.7 in \cite{Brotherston11}) says
that the provability $\LKID$ includes that of $\CLKID$.
Simpson \cite{Simpson-2017} proved the conjecture in the case of Peano Arithmetic.
The goal of this paper is to prove that the conjecture is false in general, by showing that there is no proof of $2$-Hydra in $\LKID(\Sigma_N,\Phi_N)$.

\section{The Structure $\setM$ for the Language $\Sigma_\Npredicate$}
\label{section:counter-model}

In this section we define a structure $\setM$ for the language
$\Sigma_\Npredicate$, we prove that $\setM$ falsifies the $2$-Hydra
statement $H$, and we characterize the subsets of $\setM$ which
satisfy the induction schema
(definition \ref{definition-induction-schema}). $\setM$ is not a standard
model of $\LKID$ (in any standard model $2$-Hydra would be true). In
the next sections we will prove that $(\setM,\setH_\setM)$ is a Henkin
model of $\LKID$, where $\setH_\setM$ was defined as the set of definable
sets in $\setM$ (definition \ref{definition-definable-predicates}).

\subsection{Outline for Proof of Non-Provability}

In this section we define a counter model $\setM$, whose predicates are the equality relation and a partial bijection relation (a one-to-one correspondence between some subsets of the universe for $\setM$).
We prove that $(\setM,\setH_\setM)$ 
is a Henkin model of $\LKID(\Sigma_\Npredicate, \Phi_\Npredicate)$ if we take $\setH_\setM$ as the set of definable sets of the theory of $\setM$  (definition \ref{definition-definable-predicates}). In fact, we will prove that $\setM$ satisfies the induction schema for $\Npredicate$  (definition \ref{definition-induction-schema}). 

On one hand, we prove that in our structure $\setM$ all definable sets of $\setM$ (that is, all unary definable predicates of $\setM$) are all sets we obtain by adding/removing finitely many elements to: the set $\universe{\setM}$, the set of even numbers in $\universe{\setM}$, the set of multiples of four in $\universe{\setM}$, and so forth, together with their translations and their finite unions. All these sets have measure of a dyadic rational, which is a rational of the form $z/2^m$ for some $z \in \Zeta$, $m\in \Nat$.
This claim about the measure of definable sets of $\setM$ is derived as a particular case of a 
quantifier-elimination result (section \ref{section:quantifier-elimination}), in
which we characterize the sets which are first order definable from a set of
partial bijections closed under composition and inverse (section
\ref{section:partial-bijections}). This result is new, as far as we know. For an
introduction to quantifier-elimination we refer to~\cite[section 3.1, section 3.2]{Enderton00}.

On the other hand, section \ref{subsection:measure} shows
that a definable set of $\setM$ with dyadic measure satisfies the induction schema for $\Npredicate$ (definition \ref{definition-induction-schema}).
Combining them, finally we will show that
$\setM$ satisfies the induction schema for $\Npredicate$ and therefore $(\setM,\setH_\setM)$ is 
a Henkin model of $\LKID(\Sigma_\Npredicate, \Phi_\Npredicate)$.

\subsection{Definition of the Structure $\setM$}
\label{subsection:counter-model}
Let $\Zeta$ be the set of integers. We first define the structure
$\setM$ for $0,S,p$ and $\Npredicate$. 
Later we will define the interpretation of
the predicate $p$ in $\setM$. We denote the universe of $\setM$ by
$\universe{\setM}$ and we set:

\begin{defi}[The sets $\universe{\setM}$ and $\Npredicate_\setM$]
\label{definition-Npredicate}\mbox{}
\begin{enumerate}
\item
$\universe{\setM}=\Nat + \Zeta = \{(1,x) \mid  x \in \Nat\} \cup \{(2,x) \mid  x \in \Zeta\}$ (the disjoint union of $\Nat$ and $\Zeta$).
\item
$0_\setM = (1,0)$ and ${\s}_\setM(x,y) = (x,y+1)$.
\item
$\Npredicate_\setM = \universe{\setM}$. 
\end{enumerate}
\end{defi}
For all $n \in \Nat$ we set: $(x,y) + n = (x,y+n)$,
and $\zeroZ = (2,0)$ (the element $0$ in the component $\Zeta$),
and $\zeroZ - n = (2,-n)$ (the relative integer $-n$ in the component $\Zeta$).
We define the following subsets of $\universe{\setM}$:
$\NatModel = \{0_\setM + n \mid  n \in \Nat\}$
and
$\ZetaModelMinus = \{\zeroZ - (n+1) \mid  n \in \Nat\}$
and
$\ZetaModelPlus = \{\zeroZ + n \mid  n \in \Nat\}$.
The sets $\NatModel$, $\ZetaModelMinus$, $\ZetaModelPlus$ are a partition of $\universe{\setM}$.

By construction $\setM$ satisfies the closedness of $\Npredicate$ under $0$ and $\s$, and the $(0,{\s})$-axioms. $\setM$ is not a standard model of $\LKID$ because the intersection of all subsets of $\universe{\setM}$ closed under under $0$ and $\s$ is $\NatModel \subset \universe{\setM}$, while $\Npredicate_\setM = \universe{\setM}$.

Even if $\setM$ is not a standard model of $\LKID$, one can extend $\setM$ to a Henkin model $(\setM,\setH)$ of $\LKID(\Sigma_N, \Phi_N)$, provided that one can identify a suitable Henkin class $\setH$ of sets in $\setP(\universe{\setM})$
that meets the Henkin closure conditions, and show that $\Npredicate_\setM$ is the
least prefixed point of $\phi_{\Phi_N}$ within
this class. We do it by adding to $\setM$ the interpretation $p_\setM$ of the binary predicate $p$, and taking $\setH$ to be $\setH_\setM$, the set of all sets first order definable from $0, \s, =, \Npredicate, p$  (definition \ref{definition-definable-predicates}). 
We first define a (non-empty) set of points in which the instances of $2$-Hydra will be false in $\setM$. Let $r = \{(n,2n) \mid  n \in \Nat \}$. $r$ is the set of points of the straight line $y = 2x$ which are in $\Nat \times \Nat$. We imagine $r$ starting from the infinity, moving at each step from some $({\s}a,\s \s b)$ to $(a,b)$, and ending at $(0,0)$.
Given $(m_1, m_2) \in \universe{\setM} \times \universe{\setM}$ we define $(m_1, m_2) + r = \{(m_1 + a, m_2 + b) \mid  (a,b) \in r\}$ and $(m_1, m_2) - r = \{(m_1 - a, m_2 - b) \mid  (a,b) \in r\}$. We define three paths in $\universe{\setM} \times \universe{\setM}$ by $\pi_1 = (0_\setM, \zeroZ) + r $ and $\pi_2 = (\zeroZ,0_\setM) + r $  and $\pi_3 =(\zeroZ-1, \zeroZ-2) - r$. Then $\pi_1 \cup \pi_2 \cup \pi_3$ is the set of points in which $2$-Hydra will be false in the model.

Informally, the reason is that we can move forever along $\pi_1 \cup \pi_2 \cup \pi_3$ 
while ``cutting heads'', and we never reach a winning condition. Here is an example, where we write $\mapsto$ for a single move of the game, and we use the clause $H_c$ for a head duplication. The infinite sequence of moves is: 
\\
$\ldots \mapsto (0_\setM+2, \zeroZ +4) \mapsto (0_\setM+1, \zeroZ+2) \mapsto  (0_\setM, \zeroZ) \mapsto $ \emph{(head duplication, by the clause $H_c$)} $ (\zeroZ -1, \zeroZ - 2) \mapsto (\zeroZ - 2, \zeroZ - 4) \mapsto  \ldots $. In the next figure we represent $\pi_1 \cup \pi_2 \cup \pi_3$ in $\universe{\setM} \times \universe{\setM}$:

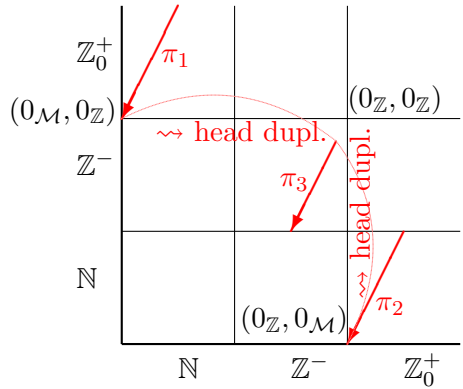
\begin{center}

\setlength{\unitlength}{1.5cm}

\begin{picture}(3,3)

\definecolor{darkred}{rgb}{0.5,0,0}
\definecolor{lightred}{rgb}{1,0.5,0.5}

{\bf
\put(2,2.1){$(\zeroZ,\zeroZ)$}
\put(-1,2){$(0_\setM,\zeroZ)$}
\put(1.05,0.15){$(\zeroZ,0_\setM)$}

\put(-0.4,0.5){$\Nat$}
\put(-0.4,1.5){$\Zeta^{-}$}
\put(-0.4,2.5){$\Zeta^{+}_{0}$}

\put(0.5,-0.3){$\Nat$}
\put(1.5,-0.3){$\Zeta^{-}$}
\put(2.5,-0.3){$\Zeta^{+}_{0}$}
}

\thinlines
\put(1,0){\line(0,1){3}}
\put(2,0){\line(0,1){3}}
\put(0,1){\line(1,0){3}}
\put(0,2){\line(1,0){3}}

\thicklines
\put(0,0){\line(1,0){3}}
\put(0,0){\line(0,1){3}}

\color{red}
\put(0.5,3){\vector(-1,-2){0.5}}
\put(2.5,1){\vector(-1,-2){0.5}}
\put(1.9,1.8){\vector(-1,-2){0.4}}

\color{lightred}
\thinlines
\qbezier[160](0,2)(1,2.5)(1.9,1.8)
\qbezier[160](2,0)(2.5,1)(1.9,1.8)

\color{red}
\put(0.35,2.5){$\pi_1$}
\put(2.25,0.3){$\pi_2$}
\put(1.40,1.4){$\pi_3$}

\put(0.3,1.8){$\rightsquigarrow$ head dupl.}

\begin{sideways}
\put(0.4,-2.2){$\rightsquigarrow$ head dupl.}
\end{sideways}

\end{picture}

\end{center}

\mbox{}
\\
Eventually, we set $$p_\setM = \universe{\setM}^2 \setminus (\pi_1 \cup \pi_2 \cup \pi_3)$$ We already defined the set $\Npredicate_{\setM}$ as $\universe{\setM}$ (definition \ref{definition-Npredicate}). We complete the definition of $\setM$ by:

\begin{defi}[The structure $\setM$]
$\setM = \langle \universe{\setM}, 0_\setM, \s_\setM, \Npredicate_\setM, p_\setM\rangle$
\end{defi}

In the following sections we will check that $\Npredicate_\setM$ is the least prefixed point of $\phi_{\Phi_\Npredicate}$ restricted to the Henkin family $\setH_\setM$ (definition \ref{definition-definable-predicates}), and therefore that $(\setM,\setH_\setM)$ is a Henkin model. In this section we check that $H$ is false in $\setM$.

\begin{lem}[The $2$-Hydra Lemma]
\label{lemma:hydra}
$\setM \not \models H$
\end{lem}

\begin{proof}[Proof]%[The $2$-Hydra Lemma \ref{lemma:hydra}]
By Def.\ \ref{definition:H}, $H = (H_a,H_b,H_c,H_d \implies \forall x,y \in \Npredicate. \ p(x,y))$. We have $\setM \not \models \forall x,y \in \Npredicate. \ p(x,y)$ because $\Npredicate_\setM = \universe{\setM}$ while $p_\setM \subset \universe{\setM}^2$. In order to prove $\setM \not \models H$, we have to prove that $\setM \models H_a,H_b,H_c,H_d$. 
\begin{enumerate}

\item
$\setM \models H_a$. We have to prove that for all $x \in \universe{\setM}$ we have: $(0_\setM, 0_\setM), (0_\setM + 1, 0_\setM), (x, 0_\setM + 1) \in p_\setM$, that is: $(0_\setM, 0_\setM), (0_\setM + 1, 0_\setM), (x, 0_\setM + 1) \not \in\pi_1 \cup \pi_2 \cup \pi_3$. For all $n,m \in \Nat$, the sets $\pi_2 \cup \pi_3$ include no point of the form $(0_\setM +n, 0_\setM+m)$: this proves $(0_\setM, 0_\setM), (0_\setM + 1, 0_\setM) \not \in\pi_1 \cup \pi_2 \cup \pi_3$. We have $(x, 0_\setM + 1) \not \in \pi_1 \cup \pi_3$ because all points in $\pi_1, \pi_3$ have the second coordinate of the form $\zeroZ + z$ for some $z \in \Zeta$. We have $(x, 0_\setM + 1) \not \in \pi_2$ because all points in $\pi_2$ have the second coordinate of the form $0_\setM + 2n$ for some $n \in \Nat$. 

\item
$\setM \models H_b$. We have to prove that for all $a, b \in \universe{\setM}$ if $\setM \models p_\setM(a,b)$ then $p_\setM({\s}_\setM(a), {\s}_\setM{\s}_\setM(b))$, that is: $(a,b) \not \in \pi_1 \cup \pi_2 \cup \pi_3$ implies $({\s}_\setM(a), {\s}_\setM{\s}_\setM(b)) \not \in \pi_1 \cup \pi_2 \cup \pi_3$. By taking the contrapositive, this is equivalent to show: $({\s}_\setM(a), {\s}_\setM{\s}_\setM(b)) \in \pi_1 \cup \pi_2 \cup \pi_3$ implies $(a,b) \in \pi_1 \cup \pi_2 \cup \pi_3$. We argue by cases. 
Assume $({\s}_\setM(a), {\s}_\setM{\s}_\setM(b)) \in \pi_1$. Then $({\s}_\setM(a), {\s}_\setM{\s}_\setM(b)) = (0_\setM +n+1, \zeroZ + 2n+2)$ for some $n \in \Nat$, hence $(a,b) = (0_\setM +n, \zeroZ + 2n) \in \pi_1$. 
Assume $({\s}_\setM(a), {\s}_\setM{\s}_\setM(b)) \in \pi_2$. Then $({\s}_\setM(a), {\s}_\setM{\s}_\setM(b)) = (\zeroZ -1 - n, \zeroZ -2 - 2n)$  for some $n \in \Nat$, hence $(a,b) = (\zeroZ -1-1-n, \zeroZ -2-2-n) \in \pi_2$. 
Assume $({\s}_\setM(a), {\s}_\setM{\s}_\setM(b)) \in \pi_3$. Then $({\s}_\setM(a), {\s}_\setM{\s}_\setM(b)) = (\zeroZ +n+1, 0_\setM + 2n+2)$  for some $n \in \Nat$, hence $(a,b) = (\zeroZ +n, 0_\setM + 2n) \in \pi_3$. 

\item
$\setM \models H_c$. 
We have to prove that for all $b \in \universe{\setM}$ if $\setM \models p_\setM({\s}_\setM (b), b)$ then $p_\setM(0_\setM, {\s}_\setM{\s}_\setM(b))$, that is: $({\s}_\setM b, b) \not \in \pi_1 \cup \pi_2 \cup \pi_3$ implies $(0_\setM, {\s}_\setM{\s}_\setM(b)) \not \in \pi_1 \cup \pi_2 \cup \pi_3$. 
By taking the contrapositive, this is equivalent to show:$(0_\setM, {\s}_\setM{\s}_\setM(b)) \in \pi_1 \cup \pi_2 \cup \pi_3$ implies $({\s}_\setM (b), b) \in \pi_1 \cup \pi_2 \cup \pi_3$. 
We argue by cases. 
Assume $(0_\setM, {\s}_\setM{\s}_\setM(b)) \in \pi_1$. Then  $(0_\setM, {\s}_\setM{\s}_\setM(b)) = (0_\setM, \zeroZ)$, hence $({\s}_\setM (b), b) = (\zeroZ-1, \zeroZ-2) \in \pi_2$. 
Assume $(0_\setM, {\s}_\setM{\s}_\setM(b)) \in \pi_2 \cup \pi_3$. This cannot be, because all points in $\pi_2, \pi_3$ have the first coordinate of the form $\zeroZ + z$ for some $z \in \Zeta$.

\item
$\setM \models H_d$. 
It is similarly proved to the previous case.\qedhere
\end{enumerate}
\end{proof}

$\setH_\setM$ was defined as the set of definable sets in
$\setM$  (definition \ref{definition-definable-predicates}). 
We prove that $(\setM,\setH_\setM)$ is a Henkin model of $\LKID$.
We have to prove that $\Npredicate_\setM$ is the
smallest pre-fixed point in $\setH_1$ for $\phi_{\Phi_N}$ (definition \ref{definition:monotone-operator}). An equivalent condition is to prove the induction schema for $N$: $A[0/x],(\forall x.Nx, A \rightarrow A[x+1/x]) \rightarrow \forall x.(Nx \rightarrow A)$ for any $A \in L(\setM)$. Since the interpretation of $N$ is $\universe{\setM}$ itself, we have in fact to prove that all $X \in \setH_1$ which are closed under $0$ and $\s$ are equal to $\universe{\setM}$.

\subsection{The Measure for the Subsets of $\setM$ Closed Under $0$ and $\s$}
\label{subsection:measure}
In this subsection we define a sufficient condition for a subset of $\setM$ to satisfy the induction schema for $\Npredicate$, by using a finitely additive measure $\mu(X)$, defined on
subsets $X \subseteq \universe{\setM}$.
We will prove that
all definable subsets for $\setM$ satisfy this condition.

\begin{defi}[Measure for Subset of $\setM$]
For $X \subseteq \universe{\setM}$ we set:
$$\mu(X) = \lim_{x \rightarrow \infty} \frac{| \ \{0_\setM + n, \zeroZ-n, \zeroZ+n \in \universe{\setM} \mid  n \in [0,x] \cap N \} \cap X \ |}{3(x+1)}$$
whenever this limit exists.
\end{defi}

For instance, $\mu(\NatModel) = 1/3$ and if $E = \{0_\setM, 0_\setM + 2, \ldots, \zeroZ -2, \zeroZ, \zeroZ + 2, \ldots\}$, then $\mu(E) = 1/2$.
A \emph{dyadic rational} is any rational of the form $z/2^n$ for some $z \in \Zeta$, $n \in \Nat$. We prove that having a dyadic measure is a sufficient condition for a predicate $A[x]$ to satisfy the induction schema for $\Npredicate$, namely: $A[0], \forall x.\Npredicate x,A[x] \rightarrow A[\s x]) \rightarrow \forall x.(\Npredicate x \rightarrow A[x])$ (definition \ref{definition-induction-schema}). Later, we will prove that all definable predicates of $\setM$ have a dyadic measure, hence they satisfy the induction schema.

\begin{lem}[Measure Lemma]
\label{lemma:measure}
If $\mu(P)$ is a dyadic rational, then $P$ satisfies the induction schema for $\Npredicate$.
\end{lem}

\begin{proof}

In order to show the contraposition,
assume $P$ does not satisfy the induction schema.
Then
$P$ is closed under $0$, ${\s}$ (hence $P \supseteq
\NatModel$)
and there is some $a \in \universe{\setM} \setminus
P$. From $P \supseteq \NatModel$ we deduce that $a \not \in
\NatModel$, hence $a = \zeroZ + z$ for some $z \in \Zeta$. Let $S_a =
\{a, a-1,a-2,a-3, \ldots\}$: by the contrapositive of closure under
$\s$, we deduce that $S_a \subseteq \universe{\setM} \setminus
P$. Thus, $\universe{\setM} \setminus P = \bigcup \{S_a \mid a \in
\universe{\setM} \setminus P \}$. If there is a maximum $a \in
\universe{\setM} \setminus P$ we conclude that $\universe{\setM}
\setminus P = S_a = \{\ldots, a-3,a-2,a-1,a\}$, while if there is no
maximum for $\universe{\setM} \setminus P$ then $\universe{\setM}
\setminus P = \ZetaModelMinus \cup \ZetaModelPlus$. In the first case
we have $\mu(\universe{\setM} \setminus P) = 1/3$, in the second one
we have $\mu(\universe{\setM} \setminus P) = 2/3$. Thus, if $P$ is a
counter-example to the induction schema for $\Npredicate$ then $\mu(P)
= 1/3, 2/3$ and $\mu(P)$ is not a dyadic rational.
\end{proof}

An example: if $P = \NatModel \cup {\Zeta}^+_0$, then $P$ is closed under $0$, ${\s}$ and $\zeroZ -1 \not \in P$. $P$ does \emph{not} satisfy the induction schema and $\mu(P) = 2/3$ is \emph{not} dyadic.

\section{The Set $\setR$ of Partial Bijections on $\universe{\setM}$}
\label{section:partial-bijections}

In this section we introduce some set $\setR$ of partial bijections on
$\universe{\setM}$, whose domains have some dyadic rational measure. In
sections \ref{section:quantifier-elimination}, \ref{section:main} we
will prove that all definable sets in $\setM$ (definition
\ref{definition-definable-predicates}) are domains of bijections in
$\setR$, therefore all these have dyadic rational measure, and by Lemma
\ref{lemma:measure} they satisfy the induction schema for $\Npredicate$
(definition \ref{definition-induction-schema}). 
We will conclude that $(\setM,\setH_\setM)$ is a Henkin model of
$\LKID(\Sigma_\Npredicate, \Phi_\Npredicate)$.

For a set $X$ and binary relations $R, S$ we write: $\id_X = \{(x,x) \mid  x \in X\}$, $\dom(R) = \{x \mid \exists y.(x,y) \in R\}$, $\range(R) = \{y \mid \exists x.(x,y) \in R\}$, $R^{-1} = \{(y,x)  \mid  (x,y) \in R\}$, $R \comp S = \{(x,z)  \mid  \exists y. ((y,z) \in R) \wedge ((x,y) \in S) \}$ and $R \lceil X = \{(x,y) \in R  \mid  x \in X\}$. 
Note that we write a relation composition $R \comp S$ \emph{in the same order as function composition}.

\subsection{The set $\setD$}
\label{subsection:D}

In this subsection we propose a candidate $\setD$ for the definable
subsets of $\setM$ (definition
\ref{definition-definable-predicates}). $\setD$ will consist of
$\universe{\setM}$, the set including every other elements of 
$\universe{\setM}$, the set including every other elements of 
the previous set, and so
forth. $\setD$ is closed under translations and finite unions and
adding/removing finitely many elements. We first define the
equivalence relation $\sim$, the subset $M(2^r, z)$ of
$\universe{\setM}$, and a set $\setB$ of subsets for
$\universe{\setM}$, then we will define $\setD$.

For sets $I$, $J$ we define $I\subsetsim J$ as ``$(I \setminus J)$ is finite'': this means ``$I \subseteq J$ up to finitely many elements''. We define $I \sim J$ as $I\subsetsim J \wedge J\subsetsim I$: this means ``$I, J$ are equal up to finitely many elements''.  $I \sim J$ is equivalent to: $(I \setminus J) \cup (J \setminus I)$ is finite.
For $r \in \Nat$, $s \in \Zeta$ we define the following set of elements of $\universe{\setM}$:
$$
M(2^r,s) \
= 
\{0_\setM + (2^r*z + s)  \mid  2^r*z + s \ge 0 \wedge  z \in \Zeta \} \
\cup \
\{\zeroZ + (2^r*z + s)    \mid  z \in \Zeta \}
$$  
If $r=0, s=0$ then $M(2^r,s) = \universe{\setM}$. We write
$\setB$ for the set of all sets $M(2^r,s)$, for some $r \in \Nat$, $s
\in \Zeta$. Since $2^r > 0$, all sets $M(2^r,s)$ are infinite.
We define $\setD$ as the set of subsets which equal
finite unions of sets in $\setB$ up to finitely many elements.

\begin{defi}[The set  $\setD$]
\label{definition:D}
$D \in \setD$ if and only if $D \subseteq \universe{\setM}$ and $D \sim (B_1 \cup \ldots \cup B_n)$ for some $B_1, \ldots, B_n \in \setB$. 
\end{defi}

We prove that
every set in $\setD$ has some dyadic rational measure.

\begin{lem}[$\setD$-Lemma]
\label{lemma:D}
Let $a_0, a \in \Nat$ and $D \in \setD$.
\begin{enumerate}
\item %1
All finite subsets of $\universe{\setM}$ are in $\setD$.
\item %2
For all $a \ge a_0$ there are $0 \le b_1 < \ldots < b_i < 2^a$ such that $M(2^{a_0},b) = (M(2^{a},b_1) \cup \ldots \cup M(2^{a},b_{i}))$.
\item %3
For any $D$ there are some $a$ and $0 \le b_1 < \ldots < b_i < 2^a$ such that $D \sim (M(2^a,b_1) \cup \ldots \cup M(2^a,b_i))$.
\item %4
$\mu(D)$ is some dyadic rational.
\item %4bi
$D$ satisfies the induction schema for $\Npredicate$.
\item %5
$\setD$ is closed under $\sim$, complement and finite union.
\end{enumerate}
\end{lem}

\begin{proof}[Proof]%[Lemma \ref{lemma:D} ($\setD$-Lemma)]
\begin{enumerate}
\item %1
$\emptyset$ is a finite union, therefore $\setD$ includes all $D \sim \emptyset$: that is, $\setD$ includes all finite subsets of $\universe{\setM}$.
\item %2
By repeatedly applying the equation $M(2^a,b) = M(2^a,b+2^a)$ we can assume that $0 \le b < 2^a$. Then we repeatedly apply the equation $M(2^a,b) = M(2^{a+1},b) \cup M(2^{a+1},b+2^a)$.
\item %3
Assume $D \sim (B_1 \cup \ldots \cup B_n)$. By the point $2$ above
there are $a_1, \ldots, a_n$ such that for all $a \ge a_1, \ldots,
a_n$ and all $i=1, \ldots, n$ there are $0 \le b_{i,1} < \ldots <
b_{i,n_i} < 2^a$ such that $B_i = (M(2^{a},b_{i,1}) \cup \ldots \cup
M(2^{a},b_{i,n_i}))$. It follows our claim with $a = \max(a_1,
\ldots, a_n)$.

\item %4
From the point $3$ above there are $a$ and $0 \le b_1 < \ldots < b_i < 2^{a}$ in $\Nat$ such that $D \sim (M(2^{a},b_1) \cup \ldots \cup M(2^{a},b_i))$. For any $M(2^{a},b) \subseteq \universe{\setM}$ we have $\mu(M(2^{a},b) \cap \universe{\setM}) = 1/2^a$. Since $0 \le b_1 < \ldots < b_i < 2^{a}$, 
the sets $M(2^{a},b_1), \ldots, M(2^{a},b_i)$ are pairwise disjoint. From $\mu(D)$ finite additive, we deduce that $\mu(D) = i/2^a$.

\item %4bis
By the point $4$ above and Lemma \ref{lemma:measure}, 
$D$ satisfies the induction schema for $\Npredicate$.

\item %5

By construction, $\setD$ is closed under $\sim$ and finite
union. Thus, we have to prove that if $D \in \setD$ then
$(\universe{\setM} \setminus D) \in \setD$. By the point $3$ above
there are $a$ and $0 \le b_1 < \ldots < b_i < 2^{a}$ in $\Nat$ such
that $D \sim (M(2^{a},b_1) \cup \ldots \cup M(2^{a},b_i))$. Assume
that $[0,2^a)\setminus \{b_1, \ldots, b_i\}= \{c_1, \ldots, c_j\}$:
then $(\universe{\setM} \setminus D) \sim (\universe{\setM} \setminus
M(2^a,b_1) \cup \ldots \cup M(2^a,b_i)) = (M(2^a,c_1) \cup \ldots \cup
M(2^a,c_j)) \in \setD$. By definition, we
conclude that $(\universe{\setM} \setminus D) \in \setD$.\qedhere
\end{enumerate}
\end{proof}

\subsection{The Set $\setR$ of Partial Bijections on $\universe{\setM}$}
\label{subsection:R}
In this subsection we define a set $\setR$ of partial bijections on $\universe{\setM}$ whose domains are in $\setD$.

From now on, for $n \in \Nat$ we call the relation $R^n$ the $n$-th
power of the relation $R$. $R^n$ is defined by iterating composition
$n$ times: $R^0$ is the identity relation and $R^{n+1}$ is $R^n \comp
R$. We define a negative power of a relation by $R^{-n} =
(R^{-1})^{n}$, where $R^{-1}$ denotes the inverse of $R$. Let $R_0 =
\universe{\setM}^2 \setminus p_\setM$.  Graphically, $R_0$ is the
union of the three lines we see in the image in section
\ref{subsection:counter-model}. We will define $\setR$ as the set of
binary relations on $\universe{\setM}$ which include $R_0^z$ for some $z
\in \Zeta$, 
constant addition relation,
plus all relations we obtain from those by restricting the
domain to some $D \in \setD$.

We first define some set $\setF$ of straight lines. $\setF$ is the set of maps $\phi: \Rational \rightarrow \Rational$, defined by $\phi(x) =  2^z x + r$ for some $z \in \Zeta$ and some $r \in \Rational$.
$\setF$ is closed under inverse: if $\phi(x) = 2^{z} x + r$, then $\phi^{-1}(x) = 2^{-z} x -r/2^{z}$. $\setF$ is closed under composition: if $\phi_i(x) = 2^{z_i} x + r_i$ for $i=1,2$, then $\phi_2(\phi_1(x)) = 2^{z_1 + z_2} x + (2^{z_2} r_1 + r_2)$.

Let $\Rational + \Rational = \{ (i,r)  \mid  i=1,2 \wedge r \in \Rational\}$. We extend the notations $0_\setM + n$ and $\zeroZ + z$ for $n\in\Nat, z\in\Zeta$ to the notations $0_\setM + r = (1,r) \in \Rational + \Rational $ and  $\zeroZ + r = (2,r) \in \Rational + \Rational $ for $r \in \Rational$. 

Let $\phi \in \setF$, $\phi(x) = 2^z x + r$ with $z \in \Zeta$ and $r
\in \Rational$. We say that $\phi_\calM$ is \emph{even} if $z$ even,
and that $\phi$ is \emph{odd} if $z$ is odd.  For any $\phi \in \setF$
we define a map $\phi_\setM:\dom(\phi) \rightarrow \Rational +
\Rational$. We set $\phi_\setM((i,r)) = (i,\phi(r))$ if $\phi$ is
even, and if $\phi$ is odd: $\phi_\setM((1,r)) = (2,\phi(r))$,
$\phi_\setM((2,r)) = (2,\phi(r))$ if $r<0$, $\phi_\setM((2,r)) =
(1,\phi(r))$ if $r \ge 0$. We say that $\phi$ is sign-preserving at
$0_\setM + a \in \universe{\setM}$ if $a \ge 0$ implies
$\phi(a) \ge 0$, it is sign-preserving at $ \zeroZ + b \in
\universe{\setM}$ if $b \ge 0 \Leftrightarrow \phi(b) \ge 0$. $\phi$
is sign-preserving on $E \subseteq \universe{\setM}$ if $\phi$ is
sign-preserving on all $e \in E$.

Assume $\phi(x)$ is sign-preserving on $E \subseteq \universe{\setM}$. We deduce: if $\phi$ is even, then $\phi_\setM(\NatModel \cap E) \subseteq \NatModel$ and $\phi_\setM(\ZetaModelMinus \cap E) \subseteq \ZetaModelMinus$ and $\phi_\setM(\ZetaModelPlus \cap E) \subseteq \ZetaModelPlus$; if $\phi$ is odd, then $\phi_\setM(\NatModel  \cap E) \subseteq \ZetaModelPlus$ and $\phi_\setM(\ZetaModelMinus  \cap E) \subseteq \ZetaModelMinus$ and $\phi_\setM(\ZetaModelPlus  \cap E) \subseteq \NatModel$. As a consequence, if $\phi$ is sign-preserving on $E \subseteq \universe{\setM}$, and $\phi(E) \subseteq F$, and $\psi$ is sign-preserving on $F$, then $(\psi \comp \phi)_\setM$ and $\psi_\setM \comp \phi_\setM$ coincide when restricted to $E$.

A {\em partial bijection} on $\universe{\setM}$ is a one-to-one relation between two subsets of $\universe{\setM}$. A {\em partial identity} on $\universe{\setM}$ is the identity relation on some subset of $\universe{\setM}$. We now define a set $\setR$ of partial bijections on $\universe{\setM}$ which are the restriction of $\phi_\setM$ to some $D \in \setD$, for some $\phi \in \setF$, and have range some $E \in \setD$. 
For instance, one bijection in $\setR$ is defined by $\phi(x) = 2^2 x$, with domain $\universe{\setM}$ and codomain $M(2^2 ,0)$, mapping $0_\setM +n \mapsto 0_\setM +4n$ and $\zeroZ +z \mapsto \zeroZ + 4z$.

We define ``even'' and ``odd'' bijections. They will be restrictions of
an even or odd power of the relation $R_0 = \universe{\setM}^2 \setminus p_\setM$.

We write $M$ for $|\calM|$.

\begin{defi}[The set of partial bijections $\setR$]
\label{definition:bijections}
Let $D, E \in \setD$ and $\phi \in \setF$, $\phi(x) = 2^z*x + r$. 
\begin{enumerate}

\item
$R$ is a $(D,E,\phi)$-bijection if $R$ is in $\calR$.
and $\dom(R)=D, \range(R)=E$ and
$\forall a \in D.\forall b \in \universe{\setM}.R(a,b) \Leftrightarrow
b=\phi_\calM(a)$ and $\phi$ is sign-preserving on $E$.

\item
$R$ is an even (odd) bijection if $R$ is a $(D,E,\phi)$ bijection
and $\phi$ is even (odd).
\item
$\setR$ is the set of all $(D,E,\phi)$-bijections for $D, E \in \setD$ and $\phi \in \setF$.
\end{enumerate}
\end{defi}

$R_0$, the complement of $p_\setM$, is an example of an odd
bijection, shown as follows. $R_0$ is a partial bijection.
For all $n \in \Nat$, $R_0$ maps $0_\setM + n
\mapsto \zeroZ + 2n$ and $\zeroZ - (n+1) \mapsto \zeroZ - 2(n+1)$ and
$\zeroZ + n \mapsto 0_\setM + 2n$. $R_0$ is 
the restriction of $\phi_\setM$ to $\universe{\setM}$, 
where $\phi$ is the odd map
$\phi(x) = 2^1x$.

We will prove that the definable sets of $\setM$ (definition
\ref{definition-definable-predicates}) can be expressed by the
propositional formulas whose predicates are equality and symbols for
predicates in $\setR$.

\begin{lem}[$\phi$-Lemma]

Let $\phi(x) = 2^{z_1}x + r_1$, $z_1 \in \Zeta$, $r_1 \in
\Rational$. 

\label{lemma:phi}
\begin{enumerate}
\item %0
Assume $r \in \Nat$, $z \in
\Zeta$ and $B = M(2^{r},z) \in \setB$ and $B' = M(2^{r+z_1},2^{z_1}z +
r_1)$.
If $B' \in \setB$ then $\phi_\calM(B) \sim B'$.
\item %1
If $B \in \setB$ and $\phi_\setM(B) \subsetsim \universe{\setM}$, then $\phi_\setM(B) \sim B'$ for some $B' \in \setB$
\item %2
If $D \in \setD$ and $\phi_\setM(D) \subseteq \universe{\setM}$ then $\phi_\setM(D) \in \setD$
\end{enumerate}
\end{lem}

\begin{proof}[Proof]\leavevmode
\begin{enumerate}
\item %0
From $B, B' \in \setB$ we deduce $B =(B \cap \NatModel) \cup (B \cap \ZetaModelMinus) \cup (B \cap \ZetaModelPlus)$, and the same for $B'$.  If $\phi$ is even we have $\phi_\setM(B \cap \NatModel) \sim B' \cap \NatModel$ and $\phi_\setM(B \cap \ZetaModelMinus) \sim B' \cap \ZetaModelMinus$ and $\phi_\setM(B \cap \ZetaModelPlus) \sim B' \cap \ZetaModelPlus$. If $\phi$ is odd we have $\phi_\setM(B \cap \NatModel) \sim B' \cap  \ZetaModelPlus$ and $\phi_\setM(B \cap \ZetaModelMinus) \sim B' \cap \ZetaModelMinus$ and $\phi_\setM(B \cap \ZetaModelPlus) \sim B' \cap \NatModel$. In both cases we conclude that $\phi_\setM(B) \sim B'$.

\item %1
Assume $B = M(2^r,z)$ for some $r \in \Nat$, $z \in \Zeta$ and $\phi_\setM(a) \in \universe{\setM}$ for all but finitely many $a \in B$. By definition, the set $\phi_\setM(B)$ includes elements of $\universe{\setM}$ with second coordinate $r_w = (2^{r+z_1}*w + (2^{z_1}z + r_1)) \in \Zeta$, for all but finitely many $w \in \Zeta$. If we choose two consecutive values of $w$ such that $r_w \in \Zeta$, we deduce that $2^{r+z_1} \in \Zeta$ and $ 2^{z_1}z + r_1 \in \Zeta$, hence that $r+z_1 \in \Nat$. Thus, $M(2^{r+z_1},2^{z_1}z + r_1)$ is a set of $\setB$. By the point $1$ above we have $\phi_\setM(B) \sim M(2^{r+z_1},2^{z_1}z + r_1)$. We conclude that $\phi_\setM(B) \sim B' $ for some $B' \in \setB$.

\item %2
If $D \in \setD$ then by Lemma \ref{lemma:D}, point $3$, we have $D \sim
M(2^{a_1},b_1) \cup \ldots \cup M(2^{a_i},b_i)$ for some $a_1, \ldots, a_i \in
\Nat$ and some $b_1, \ldots, b_i \in \Zeta$. From $\phi_\setM(D) \subseteq
\universe{\setM}$ we deduce $\phi_\setM(a) \in \universe{\setM}$ for all but
finitely many $a \in M(2^{a_1},b_1) \cup \ldots \cup M(2^{a_i},b_i)$. By the
point $2$ above we obtain $\phi_\setM(M(2^{a_1},b_1)) \sim B'_1, \ldots,
\phi_\setM(M(2^{a_i},b_i)) \sim B'_n$ for some $B'_1, \ldots, B'_n \in \setB$.
Thus, $\phi_\setM(D) \in \setD$. \qedhere
\end{enumerate}
\end{proof}

$\setR$ and $\setD$ satisfy the following closure properties.

\begin{lem}[Partial bijections]
\label{lemma:partial-bijections}
Assume that $R, S \in \setR$ and $D \in \setD$.
\begin{enumerate}

\item %1
$\id_D \in \setR$

\item %2
If $D \in \setD$ then $R(D) \in \setD$

\item %3
$R \comp S \in \setR$.

\item %4
$R^{-1} \in \setR$

\end{enumerate}
\end{lem}

\begin{proof}\leavevmode
\begin{enumerate}
\item
$\id_D$ is an even $(D,D,\id)$-bijection.

\item

If $R$ is a $(A,B,\phi)$-bijection, then by $\phi_\setM(A \cap D) = R(D)
\subseteq B \subseteq \universe{\setM}$ we deduce $\phi_\setM(A \cap D)
\subseteq \universe{\setM}$. By Lemma \ref{lemma:phi},
by $A \cap D \in \setD$ and $\phi_\setM(A \cap D) \subseteq \universe{\setM}$ 
we deduce $R(D) \in \setD$.

\item %3

$S$ is some $(A,B,\phi)$-bijection and $R$ is some
$(C,D,\psi)$-bijection.  If $R$ is an
even bijection, then $R$ maps $\NatModel$ in $\NatModel$, and
$\ZetaModelPlus$ in $\ZetaModelPlus$, and $\ZetaModelMinus$ in
$\ZetaModelMinus$.  If $R$ is an odd bijection, then $R$ maps
$\NatModel$ in $\ZetaModelPlus$, and $\ZetaModelPlus$ in $\NatModel$,
and $\ZetaModelMinus$ in $\ZetaModelMinus$.  The same holds for $\psi$
and $S$. Thus, $R \comp S \in \setR$ is even if both are even or both
are odd, and it is odd if one is odd and the other is even, and it is some
$(\phi^{-1}_\setM(B\cap C), \psi_\setM(B \cap C), \psi \comp
\phi)$-bijection. Here we use the fact that $B, C \in \setD$ imply $B
\cap C \in \setD$ by closure of $\setD$ under intersection, and that
$\phi^{-1} \in \setF$ and $\phi^{-1}_\setM(B \cap C) \subseteq A
\subseteq \universe{\setM}$ imply $\phi^{-1}_\setM(B\cap C) \in \setD$
by Lemma \ref{lemma:phi}, point $2$. In the same way we prove that
$\psi_\setM(B\cap C) \in \setD$.

\item %4

Assume that $R$ is a $(D,E,\phi)$-bijection. Then $R^{-1}$ is even
or odd according to what is $R$, and is a $(E,D,\phi^{-1})$-bijection.\qedhere

\end{enumerate}
\end{proof}

We write $L(\setR)$ for the first-order language generated from 
binary predicate symbols for
relations in $\setR$. Our goal is to prove that every first-order
definable set in $\setM$ is in $\setD$. Since the sets definable in
$L(\setR)$ include those definable in $\setM$, it is enough to prove
that any first-order definable set in $L(\setR)$ is in $\setD$. To
this aim, we need a quantifier-elimination result for a language
including $L(\setR)$.

\section{Quantifier Elimination Result for Partial Bijections}
\label{section:quantifier-elimination}

In this section we prove a quantifier elimination result for some set of
partial bijections, which is an abstract counterpart of the set
$\setR$ introduced in section \ref{section:partial-bijections}. The
quantifier elimination result holds when $\setR$ is closed under
composition and inverse. It is a simple, self-contained result
introducing a model-theoretical tool of some interest.  The only part
of this section which is used in the rest of the paper is the theorem
\ref{th:quantifier-constant-elim}, which will be used to characterize
the sets which
are first order definable from the set $\setR$. We take the definition
of quantifier elimination from \cite{Enderton00}, section 3.1, section
3.2.

\begin{defi}
A set $\calR$ of bijections on $U$ is called 
{\em closed under composition and inverse},
if $R,S \in \calR$ implies  $R^{-1}, R \circ S \in \calR$.
\end{defi}

\begin{lem}
Assume $R_1,R_2$ are in $\calR$ that is closed under composition and inverse.

\begin{enumerate}[label=(\arabic*)]
\item $\exists x_2R_1(x_1,x_2) \lequiv R_1^{-1} \circ R_1(x_1,x_1)$.
\item $R_1(x_1,x_2) \land R_2(x_2,x_3) \lequiv R_1(x_1,x_2) \land R_2 \circ R_1(x_1,x_3)$.
\end{enumerate}
\end{lem}

\begin{proof}\leavevmode
  \begin{enumerate}[label=(\arabic*)]
  \item $\leftarrow$ trivially holds. Since $R_1$ is a partial bijection, $\imp$ holds.
  \item $\imp$ trivially holds. We will show $\leftarrow$.
    Assume $R_1(x_1,x_2) \land R_2 \circ R_1(x_1,x_3)$.
    Then $R_1^{-1}(x_2,x_1)$.
    Then $R_2 \circ R_1 \circ R_1^{-1}(x_2,x_3)$.
    By definition of composition,
    $R_2 \circ R_1 \circ R_1^{-1}(x_2,x_3)~\imp~ 
    \exists x_1x_2'(R_1^{-1}(x_2,x_1) \land R_1(x_1,x_2') \land R_2(x_2',x_3))$.
    Since $R_1$ is a partial bijection, $x_2=x_2'$.
    Hence
    $R_2 \circ R_1 \circ R_1^{-1}(x_2,x_3) \imp R_2(x_2,x_3)$.
    Hence we have the left-hand side. \qedhere
  \end{enumerate}
\end{proof}

For a set $\calR$ of partial bijection on $U$ 
we will consider the structure $(U,\calR)$ where
$U$ is the universe and each partial bijection in $\calR$ is a binary relation.
We will also consider the theory of the structure $(U,\calR)$ where
each element $u$ in $U$ is denoted by the constant $u$ itself
and each partial bijection $R$ in $\calR$ is denoted by the predicate symbol $R$ itself.

\begin{thm}[Quantifier Elimination]\label{th:quantifier-elim}
If $\calR$ is a set of partial bijection on $U$ that is closed under composition and inverse,
the theory of the structure $(U,\calR)$ admits quantifier elimination.
\end{thm}

\begin{proof}
Let $\calU$ be the structure $(U,\calR)$.

\begin{enumerate}[label=(\arabic*)]
\item First we will show the following claim:

If $\calR$ is a set of partial bijection on $U$ that is closed under composition and inverse,
the theory of the structure $(U,\calR)$ without $=$ admits quantifier elimination.

Assume a quantifier-free formula $A$ of $L(\calU)$ is given.
Let $FV(A) = \{ x_1,\ldots,x_n \}$.
We assume $A$ is a disjunctive normal form.
We will find a quantifier-free formula $B$ of $L(\calU)$ that is equivalent to
$\exists x_n A$ in $\calU$.
Since $\exists x_n$ can be distributed over disjuncts,
we can assume $A$ does not contain disjunction.

First we replace $R(x_i,x_j)$ by $R^{-1}(x_j,x_i)$ in $A$ if $i>j$.
Then we can assume $i \le j$ for every $R(x_i,x_j)$ in $A$.

\begin{itemize}
\item 
Case 1. There is some positive $R_1(x_i,x_n)$ in $A$ such that $i<n$.

We replace every positive or negative $R_3(x_j,x_n)$, $R_4(x_n,x_j)$,
and $R_5(x_n,x_n)$ for $j < n$ except the atom $R_1(x_i,x_n)$ by
$R_1^{-1} \circ R_3(x_j,x_i)$,
$R_4 \circ R_1(x_i,x_j)$,
and $R_1^{-1} \circ R_5 \circ R_1(x_i,x_i)$ respectively.
Then atoms that contain $x_n$ is only the positive $R_1(x_i,x_n)$.
Let the result be $C \land R_1(x_i,x_n)$.
Then $\exists x_n A \lequiv C \land \exists x_n R_1(x_i,x_n)$.

Then $\exists x_n R_1(x_i,x_n) \lequiv R_1^{-1} \circ R_1(x_i,x_i)$.
Hence $\exists x_n A \lequiv C \land R_1^{-1} \circ R_1(x_i,x_i)$.
Take $B$ to be it.

\item
Case 2. Positive atoms that contain $x_n$ are only $R_1(x_n,x_n),\ldots,R_k(x_n,x_n)$.

Let $A$ be $A_1 \land R_1(x_n,x_n) \land \ldots \land R_k(x_n,x_n) \land 
\neg R'_1(x_n,x_n) \land \ldots \land \neg R'_m(x_n,x_n)\land
A_2$
where $k,m \ge 0$, $A_1$ does not contain $x_n$ and
each atom in $A_2$ is negative and contains both $x_n$ and $x_i$ for some $i<n$.

Let $X$ be $\{ x \in U  \mid  R_1(x,x), \ldots, R_k(x,x),
\neg R'_1(x,x), \ldots, \neg R'_m(x,x)  \}$.

\item
Case 2.1. $X$ is finite.

Let $X$ be $\{u_1,\ldots,u_m\}$.
Take $B$ to be $\Lor_{u \in X} A[x_n:=u]$.
Then $\exists x_n A \lequiv B$.

\item
Case 2.2. $X$ is infinite.

By moving $\exists x_n$ inside, 
$\exists x_n A \lequiv A_1 \land 
\exists x_n(R_1(x_n,x_n) \land \ldots \land R_k(x_n,x_n) 
\land \neg R'_1(x_n,x_n) \land \ldots \land \neg R'_m(x_n,x_n)
\land A_2)$.
Take $B$ to be $A_1$.
Then $\exists x_n A \lequiv B$, since we can show
$\exists x_n(R_1(x_n,x_n) \land \ldots \land R_k(x_n,x_n) 
\land \neg R'_1(x_n,x_n) \land \ldots \land \neg R'_m(x_n,x_n)
\land A_2)$ 
as follows.

Fix a negative atom $\neg R''(x_i,x_n)$ in $A_2$ where $i<n$.
Given $x_1,\ldots,x_{n-1}$, we have $|\{ x_n \in U  \mid R''(x_i,x_n)\}| \le 1$ since
$R''$ is a partial bijection.
Let $Y$ be $\{ x_n \in U  \mid  A_2 \}$.
Then $Y$ is cofinite.
Since $X$ is infinite, $X \cap Y \ne \emptyset$.
Hence
$\exists x_n(R_1(x_n,x_n) \land \ldots \land R_k(x_n,x_n) 
\land \neg R'_1(x_n,x_n) \land \ldots \land \neg R'_m(x_n,x_n)
\land A_2)$ 
is true.
\end{itemize}
\end{enumerate}

We have shown the claim.

\item Next we will show the theorem by using (1).

Assume a formula $A$ is given.
Let $\calR'$ be $\calR \cup \{ \Id_U \}$.
Then $\calR'$ is also closed under composition and inverse.
Let $A'$ be the formula obtained from $A$
by replacing $t_1=t_2$ by $\Id_U(t_1,t_2)$.
Then $A'$ is a formula in the theory of the structure $(U,\calR')$ without $=$.
By (1),
we have a quantifier-free formula $B'$ equivalent to $A'$ in 
the theory of the structure $(U,\calR')$.
Let $B$ be the formula obtained from $B'$
by replacing $\Id_U(t_1,t_2)$ by $t_1=t_2$.
Then $B$ is a quantifier-free formula in the theory of $\calU$ and $B \lequiv A$.
\end{proof}

\begin{thm}[Quantifier and Constant Elimination]\label{th:quantifier-constant-elim}\label{theorem:constant-elimination}
If $\calR$ is a set of partial bijection on $U$ that is closed under composition and inverse,
and $\Id_{\{u\}} \in \calR$ for every $u \in U$,
then 
in the theory of the structure $(U,\calR)$,
for any given formula,
there is some quantifier-free constant-free formula that is equivalent to the formula.
\end{thm}

\begin{proof}
Assume a formula $A$ is given.
Choose any variable $x$.
By the theorem \ref{th:quantifier-elim}, 
there is a quantifier-free formula $B$ that is equivalent to $A$.
In $B$, we replace $R(u_1,u_2)$ with constants $u_1,u_2$ by
$x=x$ if $R(u_1,u_2)$ is true,
and replace it by $\neg x=x$ if it is false.

In $B$, we replace $R(u,x)$ with constants $u$ by
$\Id_{\{u_1\}}(x,x)$ if $R(u,x)$ is equivalent to $x=u_1$,
and replace it by $\neg x=x$ if it is false.

In $B$, we replace $R(x,u)$ with constants $u$ in the same way as $R(u,x)$.

Let $C$ be the result.
Then $A \lequiv C$ and $C$ is a quantifier-free constant-free formula.
\end{proof}

\begin{exa}[Quantifier and constant elimination]
Our proof of the quantifier elimination results give an effective way
to transform any $A \in L(\setU)$ into some equivalent quantifier-free
$B \in L(\setR)$. We explain how this method works by
two examples. Assume $R_1, R_2, R_3 \in \setR$.
We will eliminate quantifiers in $\exists x_4.A$ 
for a given quantifier-free formula $A$ in the language $L(\setU)$,
by producing  some equivalent quantifier-free formula $B \in L(\setR)$.

\begin{itemize}
\item Example 1.
$A = (R_1(x_1,x_4) \land R_2(x_2,x_4) \land \neg R_3(x_3,x_4))$. 
First we can use $R_1(x_1,x_4)$ to eliminate $x_4$ in the other atoms,
since $R_1$ is a partial bijection.
Then $\exists x_4A$ is equivalent to
\[
\exists x_4(R_1(x_1,x_4) \land R_1^{-1} \circ R_2(x_2,x_1) \land 
\neg R_1^{-1} \circ R_3(x_3,x_1)).
\]
Next we move $\exists x_4$ inside and obtain an equivalent formula
\[
\exists x_4(R_1(x_1,x_4)) \land R_1^{-1} \circ R_2(x_2,x_1) \land 
\neg R_1^{-1} \circ R_3(x_3,x_1).
\]
Then we can replace $\exists x_4(R_1(x_1,x_4))$ by
$R_1^{-1} \circ R_1(x_1,x_1)$, since $R_1$ is a partial bijection,
to obtain an equivalent formula
\[
R_1^{-1} \circ R_1(x_1,x_1) \land R_1^{-1} \circ R_2(x_2,x_1) \land 
\neg R_1^{-1} \circ R_3(x_3,x_1),
\]
which we can take $B$ to be.
\item
Example 2.
$A = (R_1(x_1,x_3) \land R_2(x_4,x_4) \land \neg R_3(x_3,x_4))$. 
First we move $\exists x_4$ inside and obtain an equivalent formula
\[
R_1(x_1,x_3) \land \exists x_4(R_2(x_4,x_4) \land \neg R_3(x_3,x_4)).
\]
Let $X$ be $\{ x \in U  \mid  R_2(x,x) \}$.
We have two cases according to whether $X$ is finite or not.

\begin{itemize}
  \item
Case 1. $X$ is finite.

For example, we assume $X$ is $\{ u_1, u_2 \}$.
By replacing existential quantification by disjunction we obtain an equivalent formula
\[
R_1(x_1,x_3) \land (\neg R_3(x_3,u_1) \lor \neg R_3(x_3,u_2)).
\]
Since $R_3$ is a partial bijection,
$R_3(x_3,u)$ is equivalent to false or $x_3=u'$ for some $u' \in U$.
For example, we assume $R_3(x_3,u_1)$ is equivalent to $x_3=u_1'$
and
$R_3(x_3,u_2)$ is equivalent to $x_3=u_2'$.
Since $x_3=u_i'$ is equivalent to $\Id_{u_i'}(x_3,x_3)$ for $i=1,2$,
we can take $B$ to an equivalent formula
\[
R_1(x_1,x_3) \land (\neg \Id_{u_1'}(x_3,x_3) \lor \neg \Id_{u_2'}(x_3,x_3)).
\]

\item
Case 2. $X$ is infinite.

For a given $u_3 \in U$,
$|\{ x_4 \in U  \mid  \neg R_3(u_3,x_4) \}|$ is cofinite,
since $R_3$ is a partial bijection.
Since $X$ is infinite, $\exists x_4(R_2(x_4,x_4) \land \neg R_3(x_3,x_4))$
is true.
Hence we can take $B$ to be an equivalent formula
\[
R_1(x_1,x_3).
\]
\end{itemize}
\end{itemize}
\end{exa}

\section{Main Theorem}
\label{section:main}

In this section we prove that the statement $2$-Hydra is a
counterexample to the Brotherston conjecture. Assume $\setM$ is the
structure for the language $\Sigma_N$ with universe $\universe{\setM}$
defined in section \ref{section:counter-model}. 
Recall that $L(\setR)$ denoted the language having
a predicate symbol for each $R \in \setR$, where we identify $R$ and
the symbol denoting $R$. We write $L(\setM)$ for the language 
generated from 
$\Sigma_N \cup \universe{\setM}$. We consider the equality predicate $=$ a part
of any first order language.

We write $\calU$ for the
structure $(|\calM|, 0_\calM, s_\calM, N_\calM, p_\calM, \calR)$.
We write $L(\calU)$ for the first-order language generated from
$|\calM|, 0, s, N, p, \calR$.

\begin{lem}[Translation from $L(\calU)$ to $L(\calR)$]
\label{lemma:atomic-formulas}

Let $\setR$ be that introduced in section \ref{section:partial-bijections}.
Then all
atomic formulas $A \in L(\calU)$ are equivalent to some
formula $B \in L(\calR)$ in $\calU$.
\end{lem}

\begin{proof} 
Translation is defined to be homomorphic with logical connectives.
\begin{itemize}
\item
$N(t)$ is translated into true.
\item
$x=y$ is translated into $\Id_M(x,y)$.
\item
$x = s^n(0)$ is translated into $\Id_{\{(1,n)\}}(x,x)$.
\item
$x = s^n((i,a))$ for $(i,a) \in M$ is translated into $\Id_{\{(i,a+n)\}}(x,x)$.
\item
$x=s^n(y)$ is translated into $R(y,x)$ where
$R = (|\calM|, |\calM| - \{0,1,\ldots,n-1\}, \phi)$ and $\phi(x)=x+n$.
\item
$t_1=t_2$ is first translated into 
$\exists xy(x=t_1 \land y=t_2 \land x=y)$,
then translated into some formula by using above translation.
\item
$p(t_1,t_2)$ is first translated into 
$\exists xy(x=t_1 \land y=t_2 \land p(x,y))$,
then translated into some formula by using above translation.
\item
$R(t_1,t_2)$ is first translated into 
$\exists xy(x=t_1 \land y=t_2 \land R(x,y))$ for $R \in \calR$.
then translated into some formula by using above translation.\qedhere
\end{itemize}
\end{proof}

\begin{lem}\label{lemma:Rxx}
If $X=\{ x \in M \ |\ R(x,x) \}$ for some $R \in \calR$,
then $X \in \calD$.
\end{lem}
\begin{proof}
Let $R$ be $(A,B,\phi)$ where $\phi(x)=2^ax+b$.
If $(a,b)=(0,0)$, then $X=A$ and $X \in \calD$.
If $(a,b) \ne (0,0)$, $|X| \le 2$ since
$X$ is the intersection of the lines 
$\{ (x,y) \in M^2 \ |\ y=\phi(x), x \in A \}$
and the line $\{ (x,y) \in M^2 \ |\ x=y \}$.
Hence $X \in \calD$.
\end{proof}

\begin{thm}[Counterexample to the Brotherston-Simpson Conjecture]
\label{theorem:main}
Let $H$ be the formula defined in  Definition \ref{definition:H}.
Then $H$ has 
a proof in $\CLKID(\Sigma_\Npredicate, \Phi_\Npredicate)$, and 
no proof in $\LKID(\Sigma_\Npredicate, \Phi_\Npredicate) + (0,\s)$-axioms. 
\end{thm}

\begin{proof}%[Theorem \ref{theorem:main}]
The proof of $H$ in $\CLKID$ is shown in Theorem \ref{theorem:cyclic}. The non-provability of $H$ in $\LKID(\Sigma_\Npredicate, \Phi_\Npredicate) + (0,\s)$-axioms is shown as follows. Let $\setM$ be the structure defined in section \ref{subsection:counter-model}: $\setM$ falsifies $H$ by Lemma \ref{lemma:hydra}.

We defined $\setH_\setM$ as the Henkin family of definable sets in
$\setM$ (definition \ref{definition-definable-predicates}).  What is
left to be proved is: $(\setM, \setH_\setM)$ is a Henkin model of
$\LKID(\Sigma_\Npredicate, \Phi_\Npredicate)$. By definition of Henkin
model we have to prove that any prefixed point 
of the monotone operator $\phi_{\Phi_N}$ 
with the restriction in $\setH_\setM$ 
includes the
interpretation of $N$. As we already pointed out, since the
interpretation of $N$ is $\setM$ itself, this is to say that \emph{any
set in $\calH_1$ of $\setH_\setM$ which is closed under $0$ and $\s$ is
equal to $\setM$}.

Let $\setR, \setD$ be the set of relations and domains defined in
Def. \ref{definition:D}, \ref{definition:bijections}, and let $\setU =
(\universe{\setM}, \setR)$ be the partial bijection structure defined
by $\setR$. By induction on the formula we prove that all formulas $A
\in L(\setM)$ are equivalent in $\setU$ to some formula 
$B \in L(\setR)$: in the case $A$ is an atomic formula we use Lemma
\ref{lemma:atomic-formulas}, in the cases $A = \neg A_1, A_1 \vee A_2,
\exists x. A_1$ the induction hypothesis on $A_1, A_2$.

$\setD$ includes all singletons and it is closed under complement and
finite union by Lemma \ref{lemma:D}. $\setR$ is closed under
composition and inverse by Lemma \ref{lemma:partial-bijections},
points $3,4$, By Theorem \ref{theorem:constant-elimination}, each $A
\in L(\setU)$ having exactly one free variable $x$ is equivalent in
$\setU$ to some quantifier-free $B \in L(\setR)$. By replacing by $x$
each variable $y \not = x$ in $B$ we obtain a formula $B'$ equivalent
to $A$, which is a boolean combination of atoms of the form $R(x, x)$
for some $R \in \setR$, or the form $x=x$.  Assume an atom of the form
$R(x,x)$.  Then $D = \{x \in |\calM|^2 \mid R(x, x)\}$ is in $\calD$
by Lemma \ref{lemma:Rxx}.  Assume an atom of the form $x=x$. Then $\{x
\in U \mid x=x\}$ is $U \in \setD$. Thus, each atom $R(x, x)$ in $B'$
is equivalent to $x \in D$ for some $D \in \setD$. By the closure
properties for $\setD$ we deduce that $B'$ defines some set $D'$ in
$\setD$.  By Lemma \ref{lemma:D}, point $4$, $D'$ has
a dyadic measure, and by Lemma \ref{lemma:measure} 
$D'$ satisfies the induction schema.
Hence
$(\setM,\setH_\setM)$ is a Henkin
model of $\LKID(\Sigma_\Npredicate, \Phi_\Npredicate)$, as we wished
to show.
\end{proof}

\subsection{Non-Conservativity of the Inductive Definition System $\LKID$}
\label{subsection:non-conservativity}

This subsection shows a side result: non-conservativity of $\LKID$
with respect to additional inductive predicates, by giving a counterexample.

In the standard model, the truth of formula does not change
when we extend the model with inductive predicates that do not appear in the formula.
On the other hand, this is not the case for provability in
the inductive definition system $\LKID$.
Namely, a system may change the provability of a formula
even when we add inductive predicates that do not appear in the formula.
Namely, for a given system,
the system with additional inductive predicates may not be
conservative over the original system.

\begin{thm}[Non-Conservativity]\label{th:nonconservative}
There are $\Sigma_1,\Phi_1,\Sigma_2,\Phi_2$ such that
$\LKID(\Sigma_2,\Phi_2)$ is an extension of $\LKID(\Sigma_1,\Phi_1)$ and
$\LKID(\Sigma_2,\Phi_2)$ is not conservative over $\LKID(\Sigma_1,\Phi_1)$.
\end{thm}
\begin{proof}
Take $\Sigma_1$ to be $0,s,N,p$ and $\Phi_1$ to be $\Phi_N$ and
$\Sigma_2$ to be $\Sigma_1 \cup \{ \le \}$ and $\Phi_2$ to be $\Phi_1$ with
the production rules for $\le$.
Then 
the sequent $\zeroaxiom \vdash H$ is in the language of
$\LKID(\Sigma_1,\Phi_1)$ but it is not provable in $\LKID(\Sigma_1,\Phi_1)$
by Theorems \ref{theorem:main},
while it is provable in $\LKID(\Sigma_2,\Phi_2)$
by Theorem \ref{theorem:hydra-order}.
\end{proof}

\section{Conclusion}
\label{section:conclusion}
We have proved in Thm. \ref{theorem:main} that $\CLKID$, the formal system
of cyclic proofs \cite{Brotherston11} proves strictly more than
$\LKID$, Martin-L\"{o}f's formal system of inductive definitions with
classical logic. This settles an open question given in \cite{Brotherston11}. 
In Theorem \ref{th:nonconservative}, by the same counterexample
we also shows that if we add more inductive predicates to
$\LKID$ we may obtain a non-conservative extension of $\LKID$.

\subsection*{Acknowledgments}
This is partially supported by Core-to-Core Program (A. Advanced
Research Networks) of the Japan Society for the Promotion of Science.

\end{document}